\documentclass[journal,10pt]{IEEEtran}
\IEEEoverridecommandlockouts
\usepackage{caption} % 加载 caption 宏包
\usepackage{balance}
\usepackage{bm}
\usepackage{amsmath}
\usepackage{graphicx}
\usepackage{amssymb}
\usepackage{amsmath,amsthm} % Set the style of Theorem
\usepackage{mathtools} % For \MoveEqLeft
\usepackage{float}
\usepackage{enumerate}
\usepackage{epstopdf}
\usepackage{dblfloatfix}
\usepackage{caption}
\usepackage{amsfonts}
\usepackage{fancyhdr}
\usepackage{bbm}
\usepackage{cases}
\usepackage{cite}
\usepackage{breqn}
\usepackage{multirow}
\usepackage{mathrsfs}
\usepackage{algpascal}
\usepackage{algc}
\usepackage{algorithmicx}
\usepackage[ruled]{algorithm}
\usepackage{algpseudocode}
\usepackage{graphics}
\usepackage{epsfig}
\usepackage{mathrsfs} % the form of set
\usepackage{dsfont}
\usepackage{tabularx}
\usepackage{amsmath, amssymb}
\usepackage{amsthm}
\usepackage{color}

\usepackage{tikz}
\usetikzlibrary{shapes,arrows,positioning,calc}
\usepackage[colorlinks=true, linkcolor=blue, citecolor=blue, urlcolor=blue]{hyperref} % For hyperlinks

\usepackage{booktabs} % 导言区添加
\usepackage{subfigure}
\usepackage[spaceRequire=true,indLines=false]{algpseudocodex}

\floatname{algorithm}{\small Algorithm}

\algdef{SE}[DOWHILE]{Do}{doWhile}{\algorithmicdo}[1]{\algorithmicwhile\ #1}%

\ifCLASSINFOpdf
 \else
\fi

\begin{document}
\newtheoremstyle{theorem}% name
  {\topsep}% space above
  {\topsep}% space below
  {\itshape}% body font
  {}% indent amount
  {\itshape}% theorem head font
  {:}% punctuation after theorem head
  {.5em}% space after theorem head
  {\thmname{#1}\thmnumber{ #2}\thmnote{ (#3)}}% theorem head spec
\theoremstyle{theorem}
\newtheorem{theorem}{Theorem}
\newtheorem{definition}[theorem]{Definition}
\newtheorem{lemma}{Lemma}
\newtheorem{example}[theorem]{Example}
\newtheorem{Proposition}[theorem]{Proposition}
\newtheorem{Corollary}[theorem]{Corollary}

\title{Robust Resource Allocation for STAR-RIS-Assisted ISAC Networks Using Deep Learning}

\author{
Yuwen Cao\textsuperscript{},~\IEEEmembership{Member, IEEE}, Jiajun Wu\textsuperscript{}, Howard H. Yang\textsuperscript{},~\IEEEmembership{Member, IEEE}, Tomoaki Ohtsuki\textsuperscript{},~\IEEEmembership{Senior Member, IEEE}, and Tony Q. S. Quek\textsuperscript{},~\IEEEmembership{Fellow, IEEE}
\thanks{This paper was presented in part at IEEE VTC 2024-Fall\cite{10757612}.
Y. Cao and J. Wu are with the College of Information Science and Technology, Donghua University, Shanghai 201620, China (e-mail: ywcao@dhu.edu.cn,
wu18161020360@163.com). H. H. Yang is with the Zhejiang University (ZJU)–University of Illinois at Urbana–Champaign Institute (UIUC), Champaign, IL, USA, and also with the Zhejiang University, Haining, China (e-mail: haoyang@intl.zju.edu.cn). T. Ohtsuki is with the Department of Information and Computer Science, Keio University, Yokohama 223-8522, Japan (e-mail: ohtsuki@ics.keio.ac.jp). 
T. Q. S. Quek is with the Information Systems Technology and Design Pillar, Singapore University of Technology and Design, Singapore (e-mail:tonyquek@sutd.edu.sg).}
}

\markboth{}
{}

\maketitle

\begin{abstract}
% In this paper, we 
% focus on investigating the joint transmit beamforming and waveform designs in an integrated sensing and communication (ISAC) system
% assisted by a simultaneously transmitting and reflecting
% reconfigurable intelligent surface (STAR-RIS). 
In this paper, we 
focus on investigating the resource allocation in an integrated sensing and communication (ISAC) system
assisted by a simultaneously transmitting and reflecting
reconfigurable intelligent surface (STAR-RIS). 
Due to the coupled non-convex nature of the optimization problem, solving the above joint resource optimization problem while guaranteeing superior sensing performance
is challenging, especially in dynamic environment. To conquer such challenge, prior iterative optimization algorithms have been presented. However, practical implementations of these prior algorithms face notable computational challenges, e.g., the reliance on accurate global channel information 
% for the STAR-RIS-user links, which indeed is difficult to obtain in large networks or in dynamic environments, 
and the assumption of ideal phase shifts. In this context, we propose a model-driven deep learning framework for jointly optimizing the transmit beamforming and waveform in ISAC networks with low-overhead.
In addition, we adopt the Cramér-Rao lower bound (CRLB) as the sensing performance metric, and the multi-user interference (MUI) energy as the communication performance metric for performance evaluation. Besides, our framework jointly optimizes the transmit beamforming, waveform and the STAR-RIS reflection/transmission coefficient matrices by minimizing the weighted sum of MUI energy and CRLB. 
% Through joint optimization of the base station transmit waveform and the STAR-RIS reflection/transmission coefficient matrices, an optimization problem is proposed to minimize the weighted sum of MUI energy and CRLB. 
%Due to the coupled non-convex nature of the optimization problem, traditional iterative optimization algorithms face challenges of high computational complexity in practical applications, prompting the proposal of a model-driven deep learning algorithm.
Our study also considers the impact of imperfect channel information, making the research scenario more aligned with real-world applications. Simulation results validate that 
% under different system parameter configurations,
the proposed framework not only achieves higher sum-rate performance comparable to traditional iterative algorithms but also significantly reduces computational complexity in terms of algorithm execution time, as well as exhibits inherent robustness.
\end{abstract}


\section{Introduction}
As a prominent technique, integrated sensing and communication (ISAC) has emerged in the advent of beyond the fifth generation (B5G) and the sixth generation (6G) mobile networks, accentuating the convergence of communication and sensing paradigms \cite{9737357}. 
In light of invaluable and limited radio spectrum resources, ISAC is striving to revolutionize spectrum utilization by orchestrating the seamless coexistence of communication and sensing functionalities within a unified system and seeking their performance balance, thus significantly
enhancing spectral and energy efficiency \cite{9705498}. Moreover, ISAC pioneers the fusion of sensing and communication models, epitomized by the integration of sensing and communication capabilities directly into the infrastructure of base stations (BSs) \cite{7746569} and radar communication systems \cite{4012317}. However, practical implementation of the ISAC technique still faces notable challenges, such as high environmental reliance, severe multi-path fading, and limited
coverage, etc. 

On the other hand, 
% reconfigurable intelligent surface (RIS) is introduced as a promising technique to further enhance the communication performance. 
reconfigurable intelligent surface (RIS) has emerged as a transformative technology for wireless communications, enabling dynamic manipulation of the propagation environment to enhance coverage, spectral efficiency, and energy efficiency \cite{8796365,8741198,8910627}.
RIS consists of numerous low-cost passive reflective elements, each of which can independently modify the phase and/or amplitude of the incident signal. 
% However, this type of conventional-RIS (Con-RIS) has certain limitations. For instance, the location of the sensing users (SUs) or communication users (CUs) must be located on the same side of the RIS and BS, and full coverage cannot be realized.
To further improve the communication performance and provide additional degrees of freedom (DoFs), another new type of RIS which is called as the simultaneously transmitting and reflecting RIS (STAR-RIS) has emerged \cite{9690478}. Notably, STAR-RIS can be regarded as a bridge for communication users (CUs) and sensing users (SUs) when they are located at different sides of the STAR-RIS \cite{10757612}. 
Deploying STAR-RIS in the ISAC networks can customize the propagation environment, thereby improving the coverage range \cite{11145277,10601587}, enhancing signals of interest and suppressing interference \cite{10236498,10158937,10468784}. However, spectral resource optimization between communication and sensing to enhance spectral and energy efficiency in STAR-RIS-assisted ISAC networks without knowing accurate channel information and ideal phase shift remains challenge.  

% \textcolor{red}{
% Despite the extensive studies on the optimization of the transmit beamforming and
% STAR-RIS phase shift coefficients in ISAC systems, spectral resource optimization between communication and sensing taking into account the statistical characteristics of channel information has not been thoroughly investigated yet, particularly in
% STAR-RIS-assisted large networks or dynamic environment.
% % However, most of the previous studies design the transmit beamforming and waveform as well as  optimize the STAR-RIS reflection and transmission coefficient matrices in ISAC network based on accurate global channel information and ideal phase shift. 
% In addition, how to design energy splitting protocol in
% STAR-RIS-assisted ISAC networks deserves a further research.} 
Based on the above discussions, this paper focuses on improving the performance trade-off between CUs and SUs at the level of spectral resource
utilization in STAR-RIS-assisted ISAC system leveraging deep learning. In the following subsections, we first
review the relevant research works, followed by summarizing our contributions.

\subsection{Related Work}
%
% RIS has emerged as a transformative technology for wireless communications, enabling dynamic manipulation of the propagation environment to enhance coverage, spectral efficiency, and energy efficiency. 
Recent research in the field of the RIS technique as well as its variant has
% Recent research in this field has
explored various aspects, consisting of channel modeling, beamforming optimization, and practical deployments \cite{9140329,9551980,10506751,10591379,10596064}. Specifically, Di Renzo \textit{et al.} provided a comprehensive survey on RIS in \cite{9140329}, outlining its operating principles, hardware prototypes, and potential applications in 6G. The authors discussed how RIS can be integrated with other technologies, e.g., massive multiple-input multiple-output (MIMO) and millimeter-wave (mmWave) communications to address the challenges of path loss and interference. 
% This work established the theoretical foundation for RIS-aided systems and highlighted open research directions.
Pei \textit{et al.} conducted pioneering indoor/outdoor field trials to validate the performance of RIS-assisted wireless systems in \cite{9551980}. Their work demonstrated significant improvements in signal coverage and energy efficiency through adaptive beamforming, offering practical insights into real-world deployment challenges. 
% The study emphasized the importance of joint active-passive beamforming optimization and inspired subsequent experimental research. 
The authors in \cite{10506751} explored advanced RIS hardware designs, including transmissive RIS, stacked components, and dynamic metasurface antennas. The authors discussed how these architectures enable novel functionalities such as ISAC and energy harvesting. 
% This work highlighted the potential of RIS to support 6G’s diverse requirements, including ultra-reliable low-latency communications (URLLC) and aerial networks. 
In \cite{10591379}, the author proposed a two-level block-structured sparse Bayesian learning algorithm for RIS-assisted MIMO internet of thing (IoT) systems. This method efficiently estimates high-dimensional channels with limited pilot overhead, addressing the challenges of large-scale connectivity in IoT. 
% Simulation results showed improved estimation accuracy compared to traditional methods, making RIS a viable solution for enhancing IoT network performance.
In addition, active RIS (aRIS) in ultra-reliable low-latency communications (URLLC) systems was investigated in \cite{10596064}, which combines signal reflection and amplification to improve reliability and efficiency. The study demonstrated that small-scale aRIS outperforms large passive RIS in dynamic environments, providing critical insights for designing robust 6G communication systems. 
Reference \cite{10546993} investigated the energy efficiency of a satellite-to-ground communication system assisted by multiple RISs under both ideal and non-ideal propagation environments, where the ideal case assumes negligible shadowing, while the non-ideal case accounts for shadowing effects and path loss.

Despite the significant potential of RIS in wireless communications, its inherent limitations have gradually emerged. Traditional RIS can only reflect signals without active energy transmission, limiting performance in dynamic environments or scenarios with weak direct links. 
% For example, when users are at the edge of RIS coverage, reflected signals may be insufficient for reliable communication. The fixed reflection direction of RIS struggles to serve multi-directional users simultaneously and cannot penetrate obstacles, resulting in insufficient 3D spatial coverage.
For example, at the edge of RIS coverage, reflected signals may be too weak for reliable communication. The fixed reflection direction of RIS limits its ability to serve multi-directional users and cannot penetrate obstacles, leading to inadequate 3D spatial coverage.
High-dimensional channel estimation for RIS demands substantial pilot resources, reducing efficiency in large-scale IoT connectivity scenarios. To address these issues, STAR-RIS has been proposed. STAR-RIS integrates transmission and reflection functions to enable full-space signal manipulation, significantly enhancing system flexibility and coverage. The authors first propose the signal model of STAR-RIS and design three protocols in \cite{9570143}: i) power splitting (PS), ii) mode switching (MS), and iii) time switching (TS). By jointly optimizing the active beamforming of the BS and passive beamforming of STAR-RIS, it significantly reduces power consumption in multi-user communications. Experiments show that the TS protocol improves energy efficiency by over 30\% in unicast scenarios, and the MS protocol expands coverage by 50\% in multicast scenarios.

Furthermore, the interplay between STAR-RIS and ISAC has been studied in \cite{10050406,10835172}. More concretely, the authors in \cite{10050406} introduced the STAR-RIS technique for the first time. In addition, this paper presented a STAR-RIS-empowered ISAC framework. It first divides the sensing and communication spaces, and designs the sensing structure. Then, this work derived the Cramér–Rao bound (CRB) for angle-of-departure (AoD) estimation, followed by transforming it into a modified Fisher information matrix (FIM) optimization problem. By studying independent and coupled phase shift models, this work proposed a double-loop iterative algorithm and an alternating optimization algorithm, respectively. 
% Simulations verified that STAR-RIS outperforms traditional RIS. When communication requirements are low or components are sufficient, the two models show similar performance, and increasing passive components is more efficient. 
To address challenges involved in mmWave ISAC systems in high-mobility scenarios, reference \cite{10835172} proposed a STAR-RIS-assisted scatterer tracking scheme. By simultaneously receiving signals on both sides of the STAR-RIS, the Gaussian mixture probability hypothesis density (GM-PHD) filter is utilized to achieve real-time localization and classification of dynamic targets. 
% This study breaks through the sensing limitations of traditional ISAC systems in dynamic environments, providing a feasible solution for scenarios such as intelligent transportation.

Resource optimization in STAR-RIS-assisted ISAC networks exploiting deep learning has been investigated in \cite{10817769,10437059,10606184,10666006,11049876,11575614}. To be specific, the authors in \cite{10817769} proposed an artificial intelligence (AI)-enabled STAR-RIS aided MISO ISAC secure communication system, jointly designing transmit beamforming, receive filters, and STAR-RIS coefficients to maximize the long-term average security rate for legitimate users. For solving this non-convex optimization problem, deep reinforcement learning (DRL) algorithms, such as deep deterministic policy gradient (DDPG) and soft actor critic (SAC) are adopted. 
% Simulations verify the algorithms' effectiveness, and show that STAR-RIS outperforms benchmarks. 
Reference \cite{10437059} introduced a meta reinforcement learning (MRL) framework for full-duplex (FD) ISAC systems, where STAR-RIS assists simultaneous communication and sensing among the BS, CUs, and sensing target over the same time-frequency dimension. The goal is to maximize the sum-rate while mitigating self-interference in FD operations. In \cite{10606184}, the authors proposed a deep learning-based multi-agent DDPG (MADDPG) algorithm to solve the resource optimization problem in STAR-RIS-aided terahertz ISAC systems. It jointly optimizes beamforming and mobility parameters, to maximize the sum rate. 
% Simulation results show that MADDPG outperforms deep learning benchmark algorithms such as DDPG and proximal policy optimization (PPO) in terms of convergence and sum rate, verifying the advantages of combining STAR-RIS with dynamic delay alignment (DDA) modulation. 
In \cite{10666006}, the authors focused on studying large-scale connectivity and intelligence in 6G IoT networks based on federated learning (FL) and non-orthogonal multiple access (NOMA). This work utilized RIS to reconstruct signals and proposed a DRL-based method to jointly optimize device transmission power and RIS phase shifts, maximizing the minimum secrecy rate. 
The joint beamforming and resource allocation optimization for STAR-RIS-assisted sparse code multiple access
(SCMA) ISAC systems has been studied in \cite{11049876}, considering practical quality of service (QoS) and STAR-RIS constraints in the THz band. 
% To address this joint optimization problem, this work proposed a novel meta-DDPG algorithm, which can adapt to different conditions quickly.
% In addition, this work makes significant contributions to enhancing secure transmission and performance of 6G networks through the proposed DRL-based approach.  
% Reference \cite{10146432}
% proposed deep learning-based resource allocation approaches, which
% enable high-accuracy beam and power allocation with only a portion of the available time-sequential low-resolution images under channel uncertainties.
%
The authors in \cite{11575614}
proposed a memristor-based meta-learning (M-ML) framework to expedite spatial and temporal domain beam prediction in dynamic non-stationary environments.

\subsection{Contributions of Our Work}
Despite the extensive studies on the optimization of the transmit beamforming and
STAR-RIS phase shift coefficients in ISAC systems, spectral resource optimization between communication and sensing taking into account the statistical characteristics of channel information has not been thoroughly investigated yet, particularly in
STAR-RIS-assisted large networks or dynamic environment.
% However, most of the previous studies design the transmit beamforming and waveform as well as  optimize the STAR-RIS reflection and transmission coefficient matrices in ISAC network based on accurate global channel information and ideal phase shift. 
In addition, how to design energy splitting protocol in
STAR-RIS-assisted ISAC networks deserves a further research.

% In this paper, we investigate an ISAC system assisted by STAR-RIS and non-orthogonal multiple access (NOMA) technologies. 
Motivated by the above-mentioned, this paper focuses on optimizing the communication and sensing beamformers as well as the reflection and transmission coefficients of the STAR-RIS, to improve the performance trade-off of spectral resource utilization between CUs and SUs, while maximizing the achievable sum rate of the system. However, the optimization problem is a high-dimensional coupled non-convex problem, and traditional optimization algorithms face challenges of high complexity and signaling overhead. Based on the deep unfolding network (DUN) framework \cite{10757550}, \cite{10117541}, this paper presents
a model-driven deep learning framework for STAR-RIS-assisted ISAC network 
% designing a low-complexity optimization algorithm 
to achieve joint optimization of the communication and sensing beamformers, and the STAR-RIS phase shift coefficient matrix with a low complexity. 

Note that the core technical contributions of this paper can be summed up as follows:
\begin{itemize}
    \item 
    % This paper designs an ISAC system model assisted by STAR-RIS, which allows full-space coverage. 
    This paper designs an ISAC system model assisted by STAR-RIS.
    By leveraging the PS mode to cooperatively regulate the transmission and reflection coefficient matrices, it overcomes the semi-space coverage limitation of traditional RIS, enhancing the communication coverage and sensing accuracy in high-frequency scenarios, i.e., mmWave communications. 
    % enabling 360-degree full-space service capability.     
    % This significantly enhances the communication coverage and sensing accuracy in high-frequency scenarios, i.e., mmWave communications.
    \item We formulate a stochastic optimization problem of minimizing the weighted sum of communication and sensing performance metrics, i.e., the multi-user interference (MUI) energy and the Cramér-Rao lower bound (CRLB) of AoD estimation under constraints of the total BS transmit power and practical limitations of STAR-RIS elements. Notably, we formulate the stochastic optimization problem and design the proposed model-based deep learning framework under imperfect channel state information (CSI), rendering the solutions robust and applicable to real-world scenarios.
    % our formulation and optimization algorithm framework explicitly accounts for imperfect channel state information (CSI) by modeling additive Gaussian channel estimation errors, rendering the solutions robust and applicable to real-world scenarios.}
    %
    \item We design a model-based deep learning framework, transforming the iterative process of the alternating direction method of multipliers (ADMM) algorithm into a deep learning architecture for jointly optimizing transmit beamforming, waveform, and STAR-RIS reflection/transmission coefficients in ISAC systems. By converting complex operations into real-valued form via matrix transformations, the framework adapts to deep learning architecture. Imperfect CSI is introduced into training data to enhance the algorithm's robustness against channel uncertainties, addressing performance degradation of traditional methods under channel errors.
    \item Through simulation experiments, we demonstrate that the proposed algorithm exhibits low complexity. Furthermore, we verify that the proposed algorithm exhibits stronger robustness against channel uncertainty than the traditional ADMM algorithm, achieving a higher system sum rate under channel estimation errors. We also analyze the impacts of key system parameters on the achievable sum rate. Meanwhile, by comparing scenarios with RIS, without RIS, and assisted by orthogonal multiple access (OMA), we validate the effectiveness of the designed ISAC transmission system assisted by STAR-RIS. In particular, simulation results demonstrate that our model-based deep learning framework is capable of achieving an average execution time that is approximately~$40\%$ shorter than that of the existing state-of-the-art method.
\end{itemize}

The remainder of this work is organized as follows.
In Sec. \ref{suc:sm}, we introduce the system model and the signal model.
In Sec. \ref{sec:pmetrics}, we present the communication and sensing performance metrics adopted in this paper. In Sec. \ref{sec:pformulation}, we introduce a stochastic optimization problem. In Sec. \ref{sec:dladmm}, we present a model-based deep learning framework to address the above challenging problem with low overhead. Finally, numerical analysis and concluding remarks are provided in Secs. \ref{sec:simulationire} and \ref{sec:Con}, respectively.

 \begin{figure}
        \centering
\setlength{\belowcaptionskip}{-0.5cm}        \includegraphics[width=0.75\linewidth]{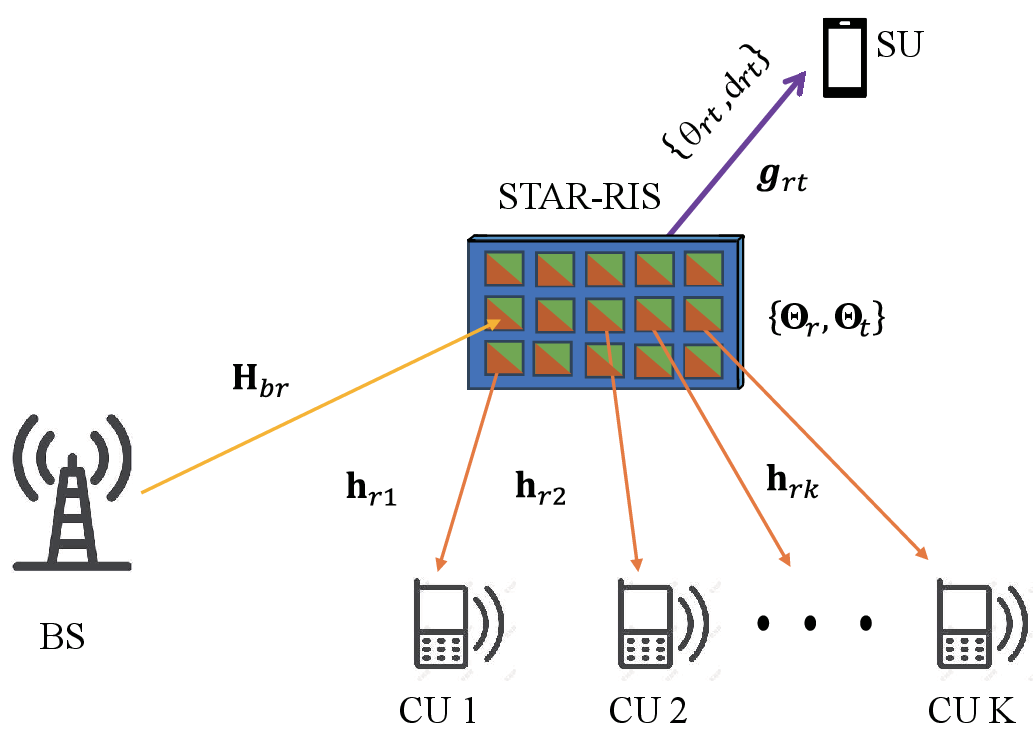}
        \caption{An overview of the STAR-RIS-assisted ISAC system model.}
        \label{fig:01}
    \end{figure}

{\it Notations}: Bold lower case letters represent vectors, while bold upper case letters
denote matrices.
$\left( {\cdot} \right)^{T}$, $\left( {\cdot} \right)^{-1}$, $\left( {\cdot} \right)^{*}$, $\left( {\cdot} \right)^{H}$, $\rm{Tr}\left( {\cdot} \right)$, $|\cdot|$, and $\left\| {\cdot} \right\|_{2}$ denote transpose, inverse, complex conjugate, Hermitian, the trace of a square matrix, the absolute value of a complex number, and $l_2$-norm operations, respectively.
%$ \oplus $ stands for the XOR operation.
$\mathbb{E}[{\cdot}]$ means the expectation of the argument. $\Re(x)$ and  $\Im(x)$ are used to denote the real part and imaginary part of variable $x$, respectively. 
$\mathbf{x}\otimes \mathbf{y}$
represents the Kronecker product between two complex vectors $\mathbf{x}$ and $\mathbf{y}$. 
$\nabla$ is used to denote the gradient of a function.
Besides, we use $\angle(\cdot{})$ to denote the argument of a complex number. In addition, vec($\mathbf{X}$) denotes the
vectorization of the matrix $\mathbf{X}$. $\text{diag}(\mathbf{x})$ returns a diagonal matrix.
% $\left\lfloor {\cdot} \right\rfloor$  indicates the flooring operator, $\left\langle  \cdot,\cdot \right\rangle $ corresponds to the inner product between two complex vectors, and $\tiny
% \left( {\begin{array}{*{20}{c}}
% {\cdot}\\
% {\cdot}
% \end{array}} \right)$ represents the binomial coefficient.

\section{System Model}
\label{suc:sm}
In this paper, we consider a nearly-passive STAR-RIS assisted mmWave communication system
that uses STAR-RIS to assist the communication and sensing among a BS, and~$K$ CUs, and an SU over the same-time frequency dimension. Fig. \ref{fig:01} illustrates the STAR-RIS-assisted ISAC system model, in which the direct communication links between the BS
and the CUs/SU are assumed to be unavailable due to blockage.\footnote{Note that our proposed model-based learning framework is valid for the scenario where the direct links between the BS and the CUs/SU exist.} 
% In our system, the STAR-RIS is equipped with $N$ near-passive scattering elements, while the BS is equipped with~$M$ transmitters/connectors.
In our system, the BS and STAR-RIS are equipped with~$M$ transmitters/connectors and~$N$ almost passive scattering elements, respectively, while each CU/SU is equipped with a single antenna.
% It can be seen from Fig. \ref{fig:01} that the BS is equipped with $M$ transmitters/connectors, STAR-RIS includes $N$ near-passive scattering elements, $K$ communication users (CU), and a single indoor sensing user (SU), which together form the STAR-RIS-assisted ISAC model. 
For convenience, herein both the STAR-RIS and the BS adopt a uniform linear array (ULA); STAR-RIS is implemented to assist in simultaneous signal reflection and refraction.

\subsection{Channel Model}
Since the STAR-RIS-assisted ISAC system in this paper operates at the mmWave band, 
% the three independent channels in the model can be constructed by using the Saleh-Valenzuela parameter channel model. 
% the independent channels between BS and STAR-RIS, between STAR-RIS and CUs, and between STAR-RIS and SU in the model can be constructed by using the Saleh-Valenzuela parameter channel model.
the independent channels between the source (BS) and the STAR-RIS, as well as between the STAR-RIS and the destination (CUs or SU), can be modeled using the Saleh-Valenzuela parameterized channel model. Specifically, as shown in Fig. \ref{fig:01}, the communication channel between STAR-RIS and CU $k$ denoted by $\mathbf{h}_{rk}\in\mathbb{C}^{N\times1}$ can be expressed as
 \begin{equation}
 \label{eq:01}
    \mathbf{h}_{rk}=\frac{e^{-\frac{j2\pi d_{rk}}{\lambda}}}{\sqrt{\rho_{rk}}}\boldsymbol{\alpha}_{{1}}(\theta_{rk}),k\in\mathcal{K}=\{1,\ldots,K\},
\end{equation}
where $d_{rk}\mathrm{~and~}\rho_{rk}$ represent the distance between STAR-RIS and CU $k$ and the path loss from the STAR-RIS to the CU $k$, respectively. For convenience, we assume in the model that $\rho_{rk}=d_{rk}^{2}$. $\lambda$ represents the 
wavelength of the carrier frequency, and $\theta_{rk}$ corresponds to the AoD of the target relative to STAR-RIS.\footnote{
In practice, we can consider the free-space path loss, which is modeled as:
\begin{align*}
    \rho_{rk}=d_{rk}^{2}f_c^2/10^{8.755}, \;  k\in\mathcal{K},
\end{align*}
where $f_c = \frac{c}{\lambda}$ (in KHz) is the carrier frequency, with $c$ being the  speed of light ($3\times 10^{8}$m/s).
} In addition, the array response vector $\boldsymbol{\alpha}_{1}(\theta_{rk})$ in Eq. (\ref{eq:01}) can be calculated as
\begin{align}
\label{eq:02}
\begin{split}
    \boldsymbol{\alpha}_{1}(\theta_{rk})=\Big[& e^{-j\frac{2\pi d_{x}}{\lambda}(\frac{N-1}{2})\mathrm{sin}(\theta_{rk})},\\&\cdots,e^{j\frac{2\pi d_{x}}{\lambda}\left(\frac{N-1}{2}\right)\mathrm{sin}(\theta_{rk})}\Big]^{T},
    \end{split}
\end{align}
where $d_{x}$ represents the spacing of inter-element, which is set to be half wavelength without loss of generality.

Similarly, the channel between STAR-RIS and SU denoted by $\mathbf{g}_{rt}$ can be expressed in the form of
\begin{equation}
\label{eq:03}
    \mathbf{g}_{rt}=\frac{e^{-j2\pi d_{rt}/\lambda}}{\sqrt{\rho_{rt}}}\boldsymbol{\alpha}_{2}(\theta_{rt}),
\end{equation}
in which $\boldsymbol{\alpha}_2(\theta_{rt})\in\mathbb{C}^{N\times1}$ is the array response vector, and in addition,
$d_{rt}$ means the distance between STAR-RIS and SU. $\rho_{rt}$ represents the path loss from the STAR-RIS to the SU, and follows the same assumption as $\rho_{rk}$, i.e., $\rho_{rt} = d_{rt}^2$. In this way, the wireless channel between the STAR-RIS and the BS denoted by $\mathbf{H}_{br}\in\mathbb{C}^{N\times M}$ can be formulated as
\begin{equation}
\label{eq:04}
    \mathbf{H}_{br}=\frac{e^{-j2\pi d_{br}/\lambda}}{\sqrt{\rho_{br}}}\boldsymbol{\alpha}_{1}(\theta_{br})(\boldsymbol{\alpha}_{2}(\theta_{br}))^{H}.
\end{equation}
Likewise, $d_{br}$ represents the distance between the BS and the STAR-RIS. For the convenience of calculation, it is assumed that $\rho_{br}=d_{br}^{2}$ represents the corresponding path loss.

In practical application scenarios, perfect channel information of ISAC networks is difficult to obtain, especially in STAR-RIS-assisted large networks or dynamic environment. Instead, to improve the feasible of practical implementation, we consider a practical application scenario where channel information
is assumed to be imperfect, thus leading to the communication and sensing performance degradation, such as limited data transmission rate, reduced target detection accuracy, etc.
% There are a series of impacts, such as limited data transmission rate, reduced target detection accuracy, etc.
% In this paper, we refer to the design of imperfect CSI in \cite{10453263}, and 
In this paper, we 
consider the \textit{additive error interference} embedded in the channel between STAR-RIS and SU denoted by $\dot{\mathbf{g}}_{rt}$, as well as in the channel between STAR-RIS and CUs denoted by $\dot{\mathbf{h}}_{rk}$, $k\in\mathcal{K}$, to simulate the imperfect CSI scenario \cite{10453263}. Accordingly, the specific channel vectors $\dot{\mathbf{g}}_{rt}$ and $\dot{\mathbf{h}}_{rk}$ can be expressed as follows:
\begin{equation}
\label{eq:05}    \dot{\mathbf{g}}_{rt}:=\mathbf{g}_{rt}+e_{rt},
\end{equation}
\begin{equation}
    \label{eq:051}\dot{\mathbf{h}}_{rk}:=\mathbf{h}_{rk}+e_{rk},
\end{equation}
where $e_{rt}$ and $e_{rk}$ follow the Gaussian distributions $\mathcal{CN}(0,\delta^2\cdot\rho_{rt})$ and $\mathcal{CN}(0,\delta^2\cdot\rho_{rk})$, respectively. Herein, $\delta^{}$ denotes the channel estimation error ratio. 
% Moreover, \textcolor{red}{$\rho_{rt}$ represents the path loss from the STAR-RIS to the SU, whereas $\rho_{rk}$ corresponds to the path loss from the STAR-RIS to the CU $k$}.

\subsection{STAR-RIS Model}
In this work, STAR-RIS operates in a PS mode, where the signal power on each element is split into two parts: one part is used to transmit the signal, and the other part is used to reflect the signal.
Let $\boldsymbol{\Theta}_{r}=\mathrm{diag}(\sqrt{\beta_{r}}e^{j\phi_{r,1}},\ldots,\sqrt{\beta_{r}}e^{j\phi_{r,N}})$ and $\boldsymbol{\Theta}_{t}=\mathrm{diag}(\sqrt{\beta_{t}}e^{j\phi_{t,1}},\ldots,\sqrt{\beta_{t}}e^{j\phi_{t,N}})$ represent the reflection and transmission coefficient matrices of STAR-RIS, respectively. In addition, $\sqrt{\beta_{r}},\sqrt{\beta_{t}}\in[0,1],\beta_{r}+\beta_{t}=1$$,n\in\mathcal{N}=\{1,...,N\}$ and $\phi_{r,n},\phi_{t,n}\in(0,2\pi]$ are the amplitude and phase shift response of the $n$-th element, respectively.

\subsection{Signal Model}
To facilitate both communication and sensing, the transmit signal of ISAC
BS during the $l$-th time slot, for $l=1,\ldots,L$, can be expressed in the form of
\begin{equation}
\label{eq:06}
    \mathbf{x}=\sum_{k=1}^K\mathbf{w}_ks_{k,l}+\mathbf{w}_ts_{t,l},
\end{equation}
where $\mathbf{w}_{k}\in\mathbb{C}^{M\times1}, k\in\mathcal{K}$, and $\mathbf{w}_{t}\in\mathbb{C}^{M\times1}$ represent the transmit beamforming matrices of the CU $k$ and SU, respectively. In addition, $s_{k,l}$ and $s_{t,l}$ correspond to the signals delivered by the BS to the $k$-th CU and the SU, respectively. In this way, the signal received by the CU $k$ during the $l$-th time slot can be mathematically expressed in the form of
\begin{equation}
\label{eq:07}
\begin{aligned}
y_{k,l}&=\dot{\mathbf{h}}_{rk}^{H}\mathbf{\Theta}_r\mathbf{H}_{br}\mathbf{w}_ks_{k,l}+\sum_{m\neq k}^K\dot{\mathbf{h}}_{rk}^{H}\mathbf{\Theta}_r\mathbf{H}_{br}\mathbf{w}_ms_{m,l}\\
&\quad+\dot{\mathbf{h}}_{rk}^{H}\mathbf{\Theta}_r\mathbf{H}_{br}\mathbf{w}_ts_{t,l} + n_{k,l},
\end{aligned}
\end{equation}
where $n_{k,l} \sim \mathcal{CN}(0,\sigma_{k}^{2})$ means the additive white Gaussian noise (AWGN) with covariance $\sigma_k^{2}$ at the $k$-th CU. Likewise, the signal observed at the sensing target can be expressed as:
\begin{equation}
  \label{eq:08} y_{s,l}=\sum_{k=1}^K\dot{\mathbf{g}}_{rt}^H\mathbf{\Theta}_t\mathbf{H}_{br}\mathbf{w}_ks_{k,l}+\dot{\mathbf{g}}_{rt}^H\mathbf{\Theta}_t\mathbf{H}_{br}\mathbf{w}_ts_{t,l}+n_{t,l},
\end{equation}
where $n_{t,l}\sim\mathcal{CN}(0,\sigma_{t}^{2})$ represents the zero-mean AWGN noise with covariance $\sigma_t^{2}$ at the SU.

\textit{Remark 1}: In traditional dual function radar communication (DFRC) systems, FD transmission is implemented to ensure that the BS having the sensing functionality can simultaneously transmit signals and receive echo signals. However, severe residual self-interference (RSI) will be involved \cite{7892029}. To mitigate such interference, there have been several works for RSI cancellation in the field of radar signal processing, which however will incur additional transmit power consumption. Based on the above analysis and for ease of practical implementation, in our signal model, the BS detects and senses the target without the need of echo signals that are often used for tracking the target in DFRC systems. That is to say, there is no reflective SU-STAR-RIS-BS link in our considered ISAC transmission system. Meanwhile, in our STAR-RIS-assisted ISAC networks, \textit{continuous phase-shifting control} strategy is adopted 
% without knowing perfect channel information is considered
for the following reasons: 
i) The discrete RIS phase shift optimization as studied in conventional discrete RIS based ISAC system, using quantization schemes \cite{9591331}, \cite{9110889}, relies heavily on perfect CSI; ii) In addition, the sensing accuracy of these discrete RIS based ISAC systems is limited to the acquisition of accurate channel information, particularly in the RIS-assisted large network or dynamic environment. iii) Prior works design the discrete phase shift configuration method. However, the number of the RIS
reflecting elements is usually limited and constrained by the
phase-shift resolution. Such approach is usually invalid for multiple RIS-assisted networks installing a large number of reflecting
elements, and cannot guarantee maximum achievable rate for
multiple RIS-aided networks.\footnote{Different from the existing studies on resource allocation for STAR-RIS-aided ISAC systems using either traditional optimization-based methods \cite{10050406,10835172} or machine learning-based approaches \cite{10817769,10437059,10606184}, \cite{11049876}, this work focuses on spectral resource allocation between communication and sensing tasks taking into account the statistical characteristics of channel information, particularly in STAR-RIS-assisted large networks or dynamic environment. In addition, to reduce the computational complexity incurred by large-scale iterations and matrix operations while maximizing the sum rate performance, we present a model-based learning algorithm, which
unfolds the iterative ADMM algorithm to a shallow and lightweight neural network.
% involves unfolding iterative ADMM algorithm into a deep learning structure to learn optimization step. 
Besides, imperfect CSI is introduced into training data to enhance our algorithm’s robustness against channel uncertainties.}

\textit{Remark 2:} 
% Typical radar sensing tasks include target detection
%  and parameter estimation.
To develop a comprehensive STAR-RIS-aided ISAC system capable of multi-task, the following section presents the joint dual-functional waveform and continuous STAR-RIS reflection and transmission coefficient
 designs for target AoD estimation and distance estimation tasks.\footnote{
To mitigate the challenge of high environmental dependence, as well as to facilitate practical deployment, 
% our ISAC system considers a realistic scenario where channel information is imperfect, which inevitably degrades both communication and sensing performance. 
we formulate a stochastic optimization problem and design a model-based deep learning framework that operates under imperfect CSI. More specifically, the proposed framework trains the DUNs on a CSI dataset under channel uncertainties, generating the waveform and STAR-RIS passive beamformer at the $t$‑th iteration according to the mapping function in Eq. (\ref{eq:51}). 
} More details are as follows.

\section{Communication and Sensing Performance Metrics}
\label{sec:pmetrics}
%
% Within the framework of the system model constructed in this paper, for the two core functional modules of communication and sensing,
In the proposed STAR-RIS-assisted ISAC network designed in this paper, to enable a better balance between the two core functional modules of communication and sensing,
the MUI energy and CRLB are adopted as the corresponding performance evaluation metrics, respectively. Specifically, the communication performance analysis takes the MUI energy as the quantitative basis, and this metric can effectively reflect the signal interference level of the system in multi-user communication scenarios. For sensing performance evaluation, the CRLB is selected as the theoretical benchmark to measure the optimal precision limit of target parameter estimation.
% thereby constructing a complete performance evaluation system.

\subsection{Communication Performance Metric}
This paper focuses on the dual-functional waveform design based on ISAC, taking the MUI energy as the communication performance metric. In communication systems, especially in multi-user communication scenarios, when multiple users transmit data over the same channel, they interfere with each other, which is referred to be as the multi-user interference. In this paper, since the dual-functional signals transmitted by the BS reach the sensing target through the same channel, for a CU $k$, $k\in \mathcal{K}$, it not only receives its own desired communication signal but also the communication signals from other CUs and the sensing signals from the SU. These signals constitute the multi-user interference for each CU $k$. As the receiver intends to accurately separate these signals, the communication performance of the system is affected. Therefore, this system adopts the MUI energy as the communication performance indicator of the system.
% \footnote{Notice that an extension of our proposed STAR-RIS-assisted ISAC framework adopting other communication performance metrics, such as spectral efficiency\cite{9690478}, energy efficiency and practical QoS, or user fairness\cite{10757612},\cite{10753571}, will be our future research topic. 
% }

Based on the received signal at CU  $ k $ given in Eq. (\ref{eq:07}), it is assumed that the desired signal $c_{\mathrm{cu},k,l}$ observed at the CU \( k \) in the $l$th time slot can be expressed as follows:
\begin{equation}
\label{eq:09}
c_{\mathrm{cu},k,l} = \dot{\mathbf{h}}_{rk}^{H} \boldsymbol{\Theta}_r \mathbf{H}_{br} \mathbf{w}_k s_{k,l}, \;k\in \mathcal{K}=\{1,\ldots,K\}.
\end{equation}
To simplify the notation, let $\dot{y}_{k,l}^{}=y_{k,l}-n_{k,l}$, $k\in \mathcal{K}$, denote the multi-user interference (excluding noise) received by the CU $k$. Then, it can be expressed as follows:
\begin{equation}
\label{eq:0011}
\begin{aligned}
\dot{y}_{k,l}^{}&=\dot{\mathbf{h}}_{rk}^{H}\mathbf{\Theta}_r\mathbf{H}_{br}\mathbf{w}_ks_{k,l}+\sum_{m\neq k}^K\dot{\mathbf{h}}_{rk}^{H}\mathbf{\Theta}_r\mathbf{H}_{br}\mathbf{w}_ms_{m,l}\\
&\quad+\dot{\mathbf{h}}_{rk}^{H}\mathbf{\Theta}_r\mathbf{H}_{br}\mathbf{w}_ts_{t,l}.
\end{aligned}
\end{equation}
As such, Eq. (\ref{eq:07}) can be re-written as follows:
\begin{equation}
\label{eq:11}
y_{k,l}=c_{\text{cu},k,l}+\underbrace{\left(\dot{y}_{k,l}^{}-c_{\text{cu},k,l}\right)}_{\text{MUI}}+n_{k,l}.
\end{equation}
Since $(\dot{y}_{k,l}^{}-c_{\text{cu},k,l})$ represents the MUI at the CU $k$, $k\in \mathcal{K}$, the signal-to-interference-plus-noise ratio (SINR) for the $k$-th CU therefore can be expressed as follows:
\begin{equation}
\label{eq:12}
\mathrm{SINR}_k=\frac{\mathbb{E}\big[|c_{\text{cu},k,l}|^2\big]}{\underbrace{\mathbb{E}\big[|\dot{y}_{k,l}^{^{}}-c_{\text{cu},k,l}|^2\big]}_{\text{MUI energy}}+\sigma_k^2}.
\end{equation}

According to \cite{10565781}, for a given channel with fixed power, minimizing the MUI energy corresponds to maximizing the achievable sum-rate of users. Additionally, the MUI energy is related to the mean squared error (MSE) between the received signal and the desired signal. Specifically, the MSE can be formulated as follows:
\begin{equation}
\label{eq:13}
\begin{aligned}\text{MSE}&:=\frac{1}{K}\biggl(\sum_{k=1}^{K} \mathbb{E}\Big[\parallel y_{k,l}-c_{\text{cu},k,l}\parallel_{2}^{2}\Big]\biggr)\\&=\underbrace{\left(\frac{1}{K}\sum_{k=1}^{K}\mathbb{E}\Big[\|y_{k,l}-c_{\text{cu},k,l}-n_{k,l}\|_{2}^{2}\Big]\right)}_{\text{MUI energy}}+\sigma_{k}^{2}\\&=\underbrace{\left(\frac{1}{K}\sum_{k=1}^{K}\mathbb{E}\Big[\|\dot{y}_{k,l}-c_{\text{cu},k,l}\|_{2}^{2}\Big]\right)}_{\text{MUI energy}}+\sigma_{k}^{2},\end{aligned}
\end{equation}
where $\text{MUI\ energy}:=\frac{1}{K}\sum_{k=1}^{K}\mathbb{E}\big[\parallel \dot{y}_{k,l}^{^{}}-c_{\text{cu},k,l}\parallel_{2}^{2}\big]$ represents the weighted sum of MUI energy for CUs. As thus, MUI energy is adopted as the communication performance metric in our paper.

Notably, in the system model each CU has a distinct AoD with respect to the STAR-RIS. One can use the delicately precoding technique, such as zero-forcing (ZF) schemes \cite{1040325},\cite{1313019}, to cancel the MUI in (\ref{eq:12}) perfectly. In this case, the criterion for designing multi-user beamformers $\mathbf{w}_i$, $i=\{1,\ldots,K\}$, and the STAR-RIS reflection coefficients $\mathbf{\Theta}_r$ has been to enforce the conditions, i.e.,
\begin{align}
    \dot{\mathbf{h}}_{ri}^{H} \boldsymbol{\Theta}_r \mathbf{H}_{br} \mathbf{w}_k = \mathbf{0}, \text{for} \;\forall \; i\in \mathcal{K}=\{1,\ldots,K\}, i\neq k.
\end{align}
The corresponding solution can achieve good performance as it completely cancels the MUI at every user. Nevertheless, we observe that: i) the above solution is sensitive to unmodeled interferences and other sources of distortion; ii) choosing the $\{\mathbf{w}_k\}$ according to the above criterion imposes a strong condition on the system configuration in terms of the number of antennas that are needed \cite{4202177}; iii) the ZF solution \cite{1040325},\cite{1313019} can lead to small per user signal-to-noise ratio (SNR) since it ignores the noise power in finding $\mathbf{w}_i$, $i=\{1,\ldots,K\}$.
For these reasons, we deduce that the resultant beamforming vectors $\mathbf{w}_i$ by using the MUI cancelation technique optimizes the SINR at each CU is not optimal relative to the beamforming vectors obtained via the sum rate optimization criterion. Therefore, in this paper we opt for the sum rate optimization criterion, which corresponds to minimizing the weighted sum of MUI energy over $\mathbf{w}_i$, to improve the communication performance of the proposed STAR-RIS-assisted ISAC system.

\textit{Remark 3:} Note that one major challenge for ISAC systems is that the constraint introduced by sensing functionality will constrain the DoF in waveform design and results in large MUI energy, thus leading to high MSE in (\ref{eq:13}) and implicitly degrading the communication performance.
Also note that, the SINR for the $k$-th CU in (\ref{eq:12}) is sensitive to the channel estimation error $e_{r,k}$.
This is because the path losses of non-line-of-sight (NLOS) BS-STAR-RIS-CU $k$ link are larger than that of the BS-STAR-RIS link, thus leading to destructive effect of the multi-path fading, and degrading the communication performance. Furthermore, one should note that since the NLOS propagation link of BS-STAR-RIS-CU $k$ is sensitive to the $e_{r,k}$, resources in terms of the \textit{temporal}- and \textit{power-domains} for these weakest channels shall be properly allocated to maintain high communication and sensing performance; Related studies will be focused in our future research.
% \footnote{Note that the potential deployment scenes of the STAR-RIS controller include the ultra-high dense urban environment, remote or coverage blind areas. These scenarios have a high demand for low-latency 6G ISAC communication. Note also that, our next step research will investigate
% task-oriented resource allocation for ISAC systems assisted by STAR-RIS, with a focus on vehicular network applications under stringent latency constraints for autonomous driving.}  

\subsection{Sensing Performance Metric}
We adopt CRLB as the sensing performance metric. The
unknown channel parameters to be estimated are those contained in $\dot{\mathbf{g}}_{rt}$, i.e., the 2-tuple vector $\boldsymbol{\nu}=[\theta_{2},d_{2}]^{T}$. Since the
additive noise is complex Gaussian distributed, by introducing
noise-free $\boldsymbol{\mu}(\boldsymbol{\nu}):= [\dot{y}_{s,1},\ldots,\dot{y}_{s,L}]^{T} = [\dot{\mathbf{g}}_{rt}^H\boldsymbol{\Theta}_t\mathbf{H}_{br}\dot{\mathbf{W}}\mathbf{s}_1,\ldots,\dot{\mathbf{g}}_{rt}^H\boldsymbol{\Theta}_t\mathbf{H}_{br}\dot{\mathbf{W}}\mathbf{s}_L]^{T} = (\dot{\mathbf{g}}_{rt}^H\boldsymbol{\Theta}_t\mathbf{H}_{br}\dot{\mathbf{W}}\mathbf{S})^{T}$ with $\dot{y}_{s,l}=\sum_{k=1}^K\dot{\mathbf{g}}_{rt}^H\mathbf{\Theta}_t\mathbf{H}_{br}\mathbf{w}_ks_{k,l}+\dot{\mathbf{g}}_{rt}^H\mathbf{\Theta}_t\mathbf{H}_{br}\mathbf{w}_ts_{t,l}$,
% \sum_{k=1}^K\dot{\mathbf{g}}_{rt}^H\boldsymbol{\Theta}_t\mathbf{H}_{br}\mathbf{w}_ks_k+\dot{\mathbf{g}}_{rt}^H\boldsymbol{\Theta}_t\mathbf{H}_{br}\mathbf{w}_ts_t$,
the FIM for $\boldsymbol{\nu}$ then is given by \cite{9364358}:
\begin{equation}
\label{eq:14}
    \mathbf{F}_{ij}:=\frac{2}{\sigma_t^{2}}\Re\{\frac{\partial\boldsymbol{\mu}^{{H}}}{\partial\boldsymbol{\nu}_{i}}\frac{\partial\boldsymbol{\mu}}{\partial\boldsymbol{\nu}_{j}}\},i,j=1,2.
\end{equation}
Herein, $\dot{\mathbf{W}} = \left[\mathbf{w}_1,\ldots,\mathbf{w}_K,\mathbf{w}_t\right]$.
In addition,
$\mathbf{S} = \left[\mathbf{s}_{_1},\ldots, \mathbf{s}_{_{L}}\right]$, and $\mathbf{s}_{_l}=[s_{1,l},\ldots,s_{K,l},s_{t,l}]^{T}$.
Specifically, the FIM matrix $\mathbf{F}$ can be written as:
\begin{equation}
\label{eq:18}
    \mathbf{F}=\begin{bmatrix}{F}_{\theta\theta}&{F}_{\theta d}\\{F}_{\theta d}^{}&{F}_{dd}\end{bmatrix},
\end{equation}
where ${F}_{\theta\theta}$, ${F}_{\theta d}$, and ${F}_{dd}$ can be respectively defined as
% \begin{equation}
% \label{eq:19}
% \begin{split}
% {F}_{\theta\theta}&=\frac{2}{\sigma_t^{2}}\Re\{(\operatorname{vec}\big(\dot{\mathbf{H}} \boldsymbol{\omega}
% \mathbf{S}\big))^{H}  \cdot  \operatorname{vec}\big(\dot{\mathbf{H}} \boldsymbol{\omega}
% \mathbf{S}\big)\}\\
% &=\frac{2}{\sigma_t^{2}} \mathrm{Tr}\Big(\mathbf{S}^{H} \boldsymbol{\omega}^{H} \dot{\mathbf{H}}^{H} \cdot \dot{\mathbf{H}} \boldsymbol{\omega}
% \mathbf{S}\Big),
% \end{split}
% \end{equation}}
\begin{equation}
\label{eq:19}
\begin{split}
{F}_{\theta\theta}&=\frac{2}{\sigma_t^{2}}\Re\{\big(\big(\dot{\mathbf{h}} \dot{\mathbf{W}}
\mathbf{S}\big)^{T}\big)^{H}  \cdot  \big(\dot{\mathbf{h}} \dot{\mathbf{W}}
\mathbf{S}\big)^{T}\}\\
&=\frac{2}{\sigma_t^{2}} \mathrm{Tr}\Big(\mathbf{S}^{T} \dot{\mathbf{W}}^{T} \dot{\mathbf{h}}^{T} \cdot \dot{\mathbf{h}}^{*} \dot{\mathbf{W}}^{*}
\mathbf{S}^{*}\Big),
\end{split}
\end{equation}

\begin{equation}
\label{eq:20}
\begin{split}
{F}_{\theta d} &=\frac{2}{\sigma_t^{2}}\Re\{((\dot{\mathbf{h}} \dot{\mathbf{W}}
\mathbf{S})^{T})^{H} \cdot (\mathbf{h} \dot{\mathbf{W}}
\mathbf{S}D)^{T}\}\\
& = -\frac{2}{\sigma_t^{2}d_{rt}} \mathrm{Tr}\Big(\mathbf{S}^{T} \dot{\mathbf{W}}^{T} {\mathbf{h}}^{T} \cdot \dot{\mathbf{h}}^{*} \dot{\mathbf{W}}^{*}
\mathbf{S}^{*}\Big),
\end{split}
\end{equation}

% \textcolor{red}{
% \begin{equation}
% \label{eq:21}
% \begin{split}
% {F}_{dd}&=\frac{2}{\sigma_t^{2}}\Re\{(\operatorname{vec}(\mathbf{H} \boldsymbol{\omega}
% \mathbf{S}D))^{H}  \cdot  \operatorname{vec}(\mathbf{H} \boldsymbol{\omega}
% \mathbf{S}D)\}\\
% &=\frac{2}{\sigma_t^{2}d_{rt}^2}\Big(1+\frac{4\pi^2d_{rt}^2}{\lambda^2}\Big) \mathrm{Tr}\Big(\mathbf{S}^{H} \boldsymbol{\omega}^{H} \mathbf{H}^{H} \cdot \mathbf{H} \boldsymbol{\omega}
% \mathbf{S}\Big),
% \end{split}
% \end{equation}}
\begin{equation}
\label{eq:21}
\begin{split}
{F}_{dd}&=\frac{2}{\sigma_t^{2}}\Re\{((\mathbf{h} \dot{\mathbf{W}}
\mathbf{S}D)^{T})^{H}  \cdot  (\mathbf{h} \dot{\mathbf{W}}
\mathbf{S}D)^{T}\}\\
&=\frac{2}{\sigma_t^{2}d_{rt}^2}\Big(1+\frac{4\pi^2d_{rt}^2}{\lambda^2}\Big) \mathrm{Tr}\Big(\mathbf{S}^{T} \dot{\mathbf{W}}^{T} \mathbf{h}^{T} \cdot \mathbf{h}^{*} \dot{\mathbf{W}}^{*}
\mathbf{S}^{*}\Big).
\end{split}
\end{equation}
% where $\mathbf{W}=[\mathbf{w}_{k},\mathbf{w}_{t}], k\in \mathcal{K}$, $\dot{\mathbf{H}}=\frac{\partial\mathbf{H}}{\partial\theta}=\mathbf{D}_{rt}\cdot\mathbf{H}$, ${D}=\frac{-\frac{j2\pi d_{rt}}{\lambda}-1}{d_{rt}}$, ${S}=[s_{k}, s_{t}]^{H}$, and $\mathbf{H}=\dot{\mathbf{g}}_{rt}^H\mathbf{\Theta}_t\mathbf{H}_{br}$. 
In addition, $\dot{\mathbf{h}}=\dot{\mathbf{g}}_{rt}^{H}\mathbf{\Theta}_{t}\mathbf{D}_{rt}^{H}\mathbf{H}_{br} \in\mathbb{C}^{1\times M}$, and $\mathbf{h} = \dot{\mathbf{g}}_{rt}^{H}\mathbf{\Theta}_{t}\mathbf{H}_{br} \in\mathbb{C}^{1\times M}$. Moreover, $D = \frac{-\frac{j2\pi d_{rt}}{\lambda}-1}{d_{rt}}$, and $\mathbf{D}_{rt}$ is defined as
% \begin{equation}
% \small
% \label{eq:17}
%     \mathbf{D}_{rt}=\mathrm{diag}\left(\left[j\pi(-\frac{N-1}{2})\cos(\theta_{rt}),\ldots,j\pi(\frac{N-1}{2})\cos(\theta_{rt})\right]\right).
% \end{equation}
\begin{equation}
\label{eq:17}
\small
\mathbf{D}_{rt}=\operatorname{diag}\Big(\Big[j\pi\big(-\frac{N-1}2\big)\cos(\theta_{rt}),\ldots,j\pi\big(\frac{N-1}2\big)\cos(\theta_{rt})\Big]\Big).
\end{equation}

Details of the
partial derivatives associated with the $\dot{\mathbf{g}}_{rt}$ parameters are given as
follows:  
% \begin{equation}
% \small
% \label{eq:15}
% \begin{split}
% \frac{\partial\mu}{\partial\theta_{rt}}&=
% \sum_{k=1}^K\dot{\mathbf{g}}_{rt}^{H}\mathbf{\Theta}_{t}\mathbf{D}_{rt}^{H}\mathbf{H}_{br}\mathbf{w}_ks_k+\dot{\mathbf{g}}_{rt}^{H}\mathbf{\Theta}_{t}\mathbf{D}_{rt}^{H}\mathbf{H}_{br}\mathbf{w}_ts_t\\
% &=\sum_{k=1}^K\dot{\mathbf{h}}\mathbf{w}_ks_k+\dot{\mathbf{h}}\mathbf{w}_ts_t,
% \end{split}
% \end{equation}
\begin{equation}
\small
\label{eq:15}
\begin{split}
\frac{\partial\boldsymbol{\mu}}{\partial\theta_{rt}}&=\Big[
\sum_{k=1}^K\dot{\mathbf{h}}\mathbf{w}_ks_{k,1}+\dot{\mathbf{h}}\mathbf{w}_ts_{t,1},\ldots,\sum_{k=1}^K\dot{\mathbf{h}}\mathbf{w}_ks_{_{k,L}}+\dot{\mathbf{h}}\mathbf{w}_ts_{_{t,L}}\Big]^{T}\\
&= \Big[\dot{\mathbf{h}}\dot{\mathbf{W}}\mathbf{s}_{_1},\ldots,\dot{\mathbf{h}}\dot{\mathbf{W}}\mathbf{s}_{_L}\Big]^{T}\\
&=\Big(\dot{\mathbf{h}}\dot{\mathbf{W}}\mathbf{S}\Big)^{T},
\end{split}
\end{equation}
% \begin{equation}
% \label{eq:16}
% \small
% \begin{split}
% \frac{\partial\mu}{\partial d_{rt}}&=D\cdot\Big(\sum_{k=1}^{K}\dot{\mathbf{g}}_{rt}^{H}\mathbf{\Theta}_{t}\mathbf{H}_{br}\mathbf{w}_ks_k+\dot{\mathbf{g}}_{rt}^{H}\mathbf{\Theta}_{t}\mathbf{H}_{br}\mathbf{w}_ts_t \Big)\\
% & = D\cdot\Big(\sum_{k=1}^{K}\mathbf{h}\mathbf{w}_ks_k+\mathbf{h}\mathbf{w}_ts_t \Big).
% \end{split}
% \end{equation}
\begin{equation}
\label{eq:16}
\small
\begin{split}
\frac{\partial\boldsymbol{\mu}}{\partial d_{rt}}
% &=D\cdot\Big(\sum_{k=1}^{K}\dot{\mathbf{g}}_{rt}^{H}\mathbf{\Theta}_{t}\mathbf{H}_{br}\mathbf{w}_ks_k+\dot{\mathbf{g}}_{rt}^{H}\mathbf{\Theta}_{t}\mathbf{H}_{br}\mathbf{w}_ts_t \Big)\\
& = D\cdot\Big[\sum_{k=1}^{K}\mathbf{h}\mathbf{w}_ks_{_{k,1}}+\mathbf{h}\mathbf{w}_ts_{_{t,1}},\ldots, \sum_{k=1}^K\mathbf{h}\mathbf{w}_ks_{_{k,L}}+\mathbf{h}\mathbf{w}_ts_{_{t,L}}\Big]^{T}\\
& = D\cdot \Big[ \mathbf{h}\dot{\mathbf{W}}\mathbf{s}_{_1},\ldots,\mathbf{h}\dot{\mathbf{W}}\mathbf{s}_{_L} \Big]^{T}\\
& = D\cdot \Big(\mathbf{h}\dot{\mathbf{W}}\mathbf{S}\Big)^{T}.
\end{split}
\end{equation}
Subsequently, we derive the CRLBs for estimating AoD and distance, which
correspond to the first and second diagonal elements of $\mathbf{F}^{-1}$, respectively, i.e.,
\begin{equation}
\label{eq:22}
    % \mathrm{CRLB}(\theta):=[\mathbf{F}^{-1}]_{1,1}=({F}_{\theta\theta}-{F}_{\theta d}{F}_{dd}{F}_{\theta d}^{T})^{-1},
    \mathrm{CRLB}(\theta):=[\mathbf{F}^{-1}]_{1,1}=({F}_{\theta\theta}-{F}_{\theta d}{F}_{dd}^{-1} {F}_{\theta d}^{})^{-1},
\end{equation}
\begin{equation}
\label{eq:0231}
    % \mathrm{CRLB}(d):=[\mathbf{F}^{-1}]_{2,2}=\left({F}_{dd}-{F}_{\theta d}{F}^{-1}_{\theta\theta}{F}_{\theta d}^T\right)^{-1}. 
    \mathrm{CRLB}(d):=[\mathbf{F}^{-1}]_{2,2}=\left({F}_{dd}-{F}_{\theta d}{F}^{-1}_{\theta\theta}{F}_{\theta d}\right)^{-1}. 
\end{equation}
\begin{lemma}
 We note that for the simplified scenario in which only a single CU and a single SU are considered in the STAR-RIS-assisted ISAC system,
 % \footnote{\textcolor{red}{Notice that the proposed model-based deep learning framework is not limited to the single sensing target scenario. An extension of the proposed STAR-RIS and NOMA-assisted ISAC transmission system to the case of multiple sensing targets will be our future research topic.}}
% Although the scheme proposed in this paper is considered for the case of one sensing target, it can be extended to the case of multiple sensing targets, which will be followed in our future research.
the corresponding $\mathrm{CRLB}(\theta)$ for estimating AoD can be derived as
\begin{align}
\small
\label{eq:23}
\begin{split}
% \mathrm{CRLB}(\theta)=\frac{\sigma_t^2}{2}\frac{1}{\mathrm{Tr}((\dot{\mathbf{H}}\mathbf{W})^H(\dot{\mathbf{H}}\mathbf{W}))-\frac{(\mathrm{Tr}((\dot{\mathbf{H}}\mathbf{W})^H(\mathbf{H}\mathbf{W}D)))^2}{\mathrm{Tr}((\mathbf{H}\mathbf{W}D)^H(\mathbf{H}\mathbf{W}D))}}.
% &\mathrm{CRLB}(\theta)=\\
% &\frac{\sigma_t^{2}}{2}\frac{1}{\mathrm{Tr}((\dot{\mathbf{H}}\boldsymbol{\omega}\mathbf{S})^{H}(\dot{\mathbf{H}}\boldsymbol{\omega}\mathbf{S}))-\frac{\lambda^2|\mathrm{Tr}((\dot{\mathbf{H}}\boldsymbol{\omega}\mathbf{S})^{H}(\mathbf{H}\boldsymbol{\omega}\mathbf{S}))|^{2}}{(\lambda^2+4\pi^2d_{rt}^2)\mathrm{Tr}((\mathbf{H}\boldsymbol{\omega}\mathbf{S})^{H}(\mathbf{H}\boldsymbol{\omega}\mathbf{S}))}}.
&{\mathrm{CRLB}}(\theta_{})=\\&\frac{\sigma_t^{2}}{2}\frac{1}{\mathrm{Tr}((\dot{\mathbf{h}}\ddot{\mathbf{W}}\dot{\mathbf{S}})^{T} (\dot{\mathbf{h}}^{} \ddot{\mathbf{W}}^{}
\dot{\mathbf{S}}^{})^{*})-\frac{\lambda^2|\mathrm{Tr}(({\mathbf{h}}\ddot{\mathbf{W}}\dot{\mathbf{S}})^{T} (\dot{\mathbf{h}}^{} \ddot{\mathbf{W}}^{}
\dot{\mathbf{S}}^{})^{*})|^{2}}{(\lambda^2+4\pi^2d_{rt}^2)\mathrm{Tr}(({\mathbf{h}}\ddot{\mathbf{W}}\dot{\mathbf{S}})^{T} ({\mathbf{h}}^{} \ddot{\mathbf{W}}^{}
\dot{\mathbf{S}}^{})^{*})}},
\end{split}
\end{align}   
where $\ddot{\mathbf{W}} = \left[\mathbf{w}_1,\mathbf{w}_t\right]$,
$\dot{\mathbf{S}} = \left[\dot{\mathbf{s}}_{_1},\ldots, \dot{\mathbf{s}}_{_{L}}\right]$, and $\dot{\mathbf{s}}_{_l}=[s_{1,l}, s_{t,l}]^{T}$.
\end{lemma}
\begin{IEEEproof}
Please see the proof in Appendix A.  
\end{IEEEproof}

\textit{Remark 4:} We note that the FIM matrix $\mathbf{F}$ derived in (\ref{eq:14}) is sensitive to the channel estimation error $e_{rt}$ as well, as the path losses of the NLOS BS-STAR-RIS-SU link are larger than that of the BS-STAR-RIS link, thus impacting the accuracy of AoD estimation.
% thus implicitly degrading the sensing performance, particularly at the low-SNR region. 
In addition, the energy splitting and the continuous phase-shifting optimization at the STAR-RIS impact the target detection and target AoD  estimation as well.  
As such, robust resource allocation is imperative to mitigate the influence of the channel estimation error, to further improve both the communication and sensing performances. 

\begin{lemma}
\label{lem02}
When taking into account the scenario where the sensing target is modeled as a surface with a large number of distributed point-like scatterers in the STAR-RIS-assisted ISAC system, such as a vehicle \cite{10138058}, a unmanned
aerial vehicle (UAV) \cite{8365881} or a pedestrian moving on the road \cite{9652071}, 
the corresponding unknown target response matrix to be estimated can be denoted by $\dot{\mathbf{G}}_{rt} \in \mathbb{C}^{N\times T}$. By introducing the noise-free $\boldsymbol{\mu}(\dot{\mathbf{G}}_{rt}):=[\dot{\mathbf{G}}_{rt}^H\boldsymbol{\Theta}_t\mathbf{H}_{br}\dot{\mathbf{W}}\mathbf{s}_1,\ldots,\dot{\mathbf{G}}_{rt}^H\boldsymbol{\Theta}_t\mathbf{H}_{br}\dot{\mathbf{W}}\mathbf{s}_L]^{T} = (\dot{\mathbf{G}}_{rt}^H\boldsymbol{\Theta}_t\mathbf{H}_{br}\dot{\mathbf{W}}\mathbf{S})^{T}$, 
% We note that for the simplified scenario in which only a single CU and a single SU are considered in the STAR-RIS and NOMA-assisted ISAC system,\footnote{\textcolor{red}{Notice that the proposed model-based deep learning framework is not limited to the single sensing target scenario. An extension of the proposed STAR-RIS and NOMA-assisted ISAC transmission system to the case of multiple sensing targets will be our future research topic.}
% % Although the scheme proposed in this paper is considered for the case of one sensing target, it can be extended to the case of multiple sensing targets, which will be followed in our future research.
% } 
the corresponding FIM with respect to $\dot{\bar{\mathbf{g}}}_{} = \text{vec}(\dot{\mathbf{G}}_{rt}^{})$ then can be derived as
% \begin{equation}
% \label{eq:23}
% \mathrm{CRLB}(\dot{\mathbf{G}}_{rt}) = \sum_{k=1}^{K}\frac{1}{\sigma_{t}^{2}}\mathbf{v}_{k}^{*}\mathbf{v}_{k}^{T}\otimes \mathbf{I}_{n}+\frac{1}{\sigma_{t}^{2}}\mathbf{v}_{t}^{*}\mathbf{v}_{t}^{T}\otimes \mathbf{I}_{n},
% \end{equation}}
% \begin{equation}
% \label{eq:0026}
% \mathbf{F}(\bar{\mathbf{g}}_{rt}) = \sum_{k=1}^{K}\frac{1}{\sigma_{t}^{2}}\boldsymbol{\upsilon}_{k}^{*}\boldsymbol{\upsilon}_{k}^{T}\otimes \mathbf{I}_{L}+\frac{1}{\sigma_{t}^{2}}\boldsymbol{\upsilon}_{t}^{*}\boldsymbol{\upsilon}_{t}^{T}\otimes \mathbf{I}_{L},
% \end{equation}
\begin{equation}
\label{eq:0026}
\mathbf{F}(\dot{\bar{\mathbf{g}}}_{}) = \frac{1}{\sigma_t^2}
    \Big(\mathbf{I}_{T} \otimes \mathbf{R}_{{\Lambda}}\Big),
\end{equation}
% for estimating the extended target parameter $\dot{\mathbf{G}}_{rt}$ can be derived as
% \begin{equation}
% \label{eq:23}
% \mathrm{CRLB}(\dot{\mathbf{G}}_{rt}) = \sum_{k=1}^{K}\frac{1}{\sigma_{t}^{2}}\mathbf{v}_{k}^{*}\mathbf{v}_{k}^{T}\otimes \mathbf{I}_{n}+\frac{1}{\sigma_{t}^{2}}\mathbf{v}_{t}^{*}\mathbf{v}_{t}^{T}\otimes \mathbf{I}_{n},
% \end{equation}}
where $\mathbf{R}_{{\Lambda}} := \boldsymbol{\Lambda} \boldsymbol{\Lambda}_{}^{H} = (\boldsymbol{\Theta}_{t}\mathbf{H}_{br}\dot{\mathbf{W}}\mathbf{S})(\boldsymbol{\Theta}_{t}\mathbf{H}_{br}\dot{\mathbf{W}}\mathbf{S})^H$. In addition, $\mathbf{I}_T$ is an $T\times T$ identity matrix.
\end{lemma}
\begin{IEEEproof}
Please see the proof in Appendix B.  
\end{IEEEproof}

\begin{lemma}
\label{lem:lem03}
The corresponding $\mathrm{CRLB}(\dot{\bar{\mathbf{g}}})$ for estimating the extended target parameter $\dot{\bar{\mathbf{g}}}_{}$ then can be derived as
\begin{equation}
\label{eq:0027}
\mathrm{CRLB}(\dot{\bar{\mathbf{g}}}) =
\mathrm{Tr\Big(\mathbf{F}^{-1}(\dot{\bar{\mathbf{g}}}_{}) \Big)}= T \sigma_{t}^{2} \cdot \mathrm{Tr\Big(\mathbf{R}_{{\Lambda}}^{-1}\Big)}.
% = \left((\mathbf{I}_{T} \otimes \boldsymbol{\Lambda}^{*}\boldsymbol{\Lambda}_{}^{T})\mathbf{A}\right)^{-1}.
\end{equation}
\end{lemma}
\begin{IEEEproof}
For the above estimators of the target response parameters (unbiased $\hat{\boldsymbol{\mu}}(\dot{\mathbf{G}}_{rt})$), we calculate the $\mathrm{CRLB}$ on the error covariance matrix as follows: 
\begin{align}
\label{eq:0028}
    \mathbb{E}\big[(\boldsymbol{\mu}-\hat{\boldsymbol{\mu}}(\dot{\mathbf{G}}_{rt}))(\boldsymbol{\mu}-\hat{\boldsymbol{\mu}}(\dot{\mathbf{G}}_{rt}))^{H}\big]\succeq\mathbf{F}^{-1}(\dot{\bar{\mathbf{g}}}_{}) ,
\end{align}
where the notation $\mathbf{X} \succeq \mathbf{Y}$ for square matrices $\mathbf{X}$ and $\mathbf{Y}$ means $\mathbf{x}^{H}\mathbf{X}\mathbf{x} \geq \mathbf{x}^{H}\mathbf{Y}\mathbf{x}$ for any valid vector $\mathbf{x}$. Then, substituting the FIM with respect to $\dot{\bar{\mathbf{g}}}_{}$ derived in (\ref{eq:0026}) into the above formula (\ref{eq:0028}) yields the corresponding $\mathrm{CRLB}(\dot{\bar{\mathbf{g}}})$ for estimating the extended target parameter in (\ref{eq:0027}). This completes the proof.
\end{IEEEproof}

\textit{Remark 5:} By inspection, we observe that if $\text{rank}(\mathbf{R}_{{\Lambda}}) < N$, then the FIM $\mathbf{F}(\dot{\bar{\mathbf{g}}}_{})$ in (\ref{eq:0026}) for estimating the target response matrix $\dot{\mathbf{G}}_{rt}$ is a singular matrix, thus leading to $\mathrm{CRLB}(\dot{\bar{\mathbf{g}}})  = \infty$. Otherwise, if $\text{rank}(\mathbf{R}_{{\Lambda}}) = N$, then the FIM $\mathbf{F}(\dot{\bar{\mathbf{g}}}_{})$ in (\ref{eq:0026}) is invertible, thereby yielding bounded  $\mathrm{CRLB}(\dot{\bar{\mathbf{g}}}) $ in (\ref{eq:0027}).

\section{Problem Formulation}
\label{sec:pformulation}
In this section, we propose to construct a resource allocation optimization problem for the STAR-RIS-assisted ISAC system based on adopted communication and sensing performance metrics, 
% enabling the STAR-RIS and NOMA-assisted ISAC system to achieve a balance between communication and sensing performance, thereby 
for improving the performance trade-off between CUs and SU at the level of spectral resource utilization.\footnote{Notice that integrating our proposed model-based learning approach into the integrated sensing, communication, and computing (ISCC) and satellite networks \cite{Wen2024ASO} may improve the performance trade-off of spectral resource utilization among sensing, communication, and computing tasks. In addition, our proposed scheme is applicable to the scenario with imperfect CSI. Thus, integrating our approach into satellite networks will render the solution robust and applicable to real-world scenarios.}

By jointly optimizing the BS transmission waveform and the reflection and transmission coefficient matrices of the STAR-RIS, we aim to minimize the weighted sum of the MUI energy and the CRLB, while satisfying i) the total power constraint and ii) the STAR-RIS phase shift constraint. Considering the imperfect CSI with additive interference, the channel information exhibits \textit{statistical characteristics}. Therefore, a stochastic optimization problem under imperfect CSI is formulated. The optimization problem is expressed as a minimum expectation problem, and the whole joint optimization problem can be formulated as follows:\footnote{Note that, in the system model the spatial division multiple access (SDMA) is employed to enable the multiplexing of same-frequency signals over multiple spatial paths.}
\begin{subequations}
\label{24}
\begin{equation}
\min_{\{\mathbf{w}_k\},\mathbf{w}_t,\mathbf{\Theta}_t,\mathbf{\Theta}_r}\mathbb{E}\Big[ \text{MUI energy} +\zeta\cdot\mathrm{CRLB}(\theta)\Big] \label{eq:24a}
\end{equation}
\begin{equation}
\text{s.t.} \hspace{1.4cm}\sum_{k=1}^K||\mathbf{w}_k||^2+||\mathbf{w}_t||^2\leq P_{max}, \label{eq:24b}
\end{equation}
\begin{equation}
\beta_{r}+\beta_{t}=1,\beta_{i}\in[0,1], \phi_{i,n}\in(0,2\pi], i\in\{t,r\},n\in \mathcal{N},
\end{equation}
% \begin{equation}
% i\in\{t,r\},n\in \mathcal{N}=\{1,\ldots,N\},
% \end{equation}
in which $P_{max}$ limits the maximum transmit power of the BS. In addition, the constraint condition (\ref{24}c) can be rewritten as
\begin{equation}
|\boldsymbol{\Theta}_r(n)|^2 + |\boldsymbol{\Theta}_t(n)|^2 = 1, \, \text{for} \;\forall \; n\in \mathcal{N},
\end{equation}
\end{subequations}
where $\mathbf{\Theta}_r(n)$ and $\mathbf{\Theta}_t(n)$ denote the elements of the \( n \)-th row and \( n \)-th column of the reflection and transmission coefficient matrices of the STAR-RIS, respectively. $\zeta$ denotes the trade-off coefficient between communication and sensing. A larger coefficient implies more focus on minimizing the CRLB. In the optimization problem, it can be clearly seen that there are coupling relationships among the transmit beamforming vectors, and the transmission coefficient matrix and the reflection coefficient matrix of the STAR-RIS.

    \begin{figure*}
        \centering
        \includegraphics[width=1\linewidth]{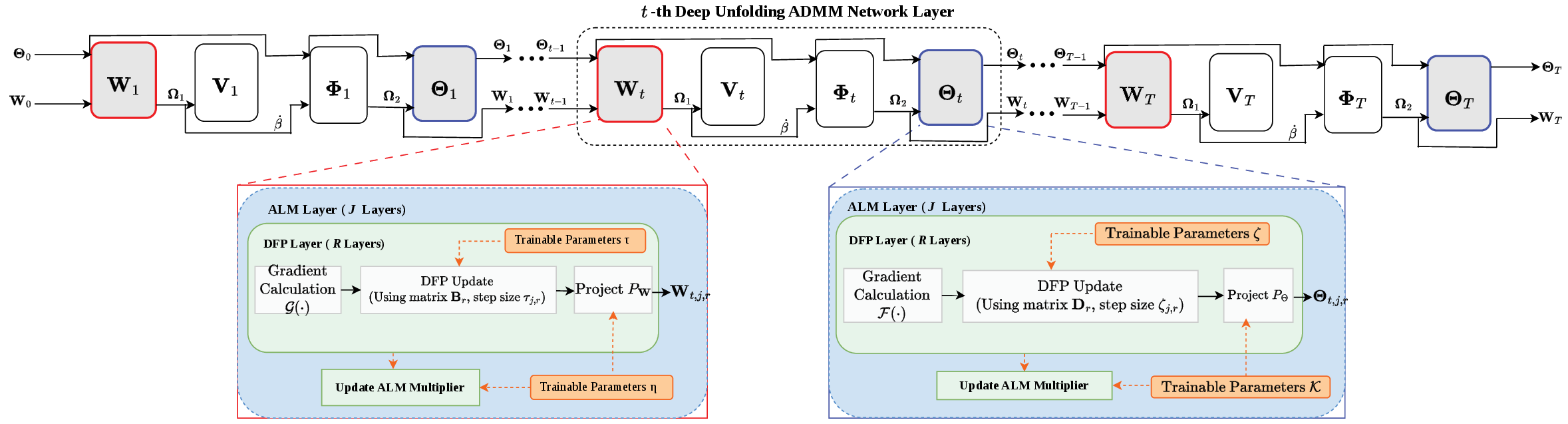}
        \caption{ Illustration of the proposed model-based learning framework.}
        \label{fig:0002}
    \end{figure*}

\section{Proposed Model-Based Deep Learning Framework}
\label{sec:dladmm}
%
% Referencing the design of model-driven deep learning algorithms in \cite{10757550}, 
In this section, we integrate the traditional mathematical optimization algorithm, i.e., ADMM \cite{10757550}, into deep learning framework. In this context, it is abbreviated as DL-ADMM, which is capable of decomposing complex problems into a series of iterative sub-problems that can be solved using unfolding-friendly algorithms. In addition, it learns the hyperparameters of the unfolding algorithm through DUN, as illustrated in Fig. \ref{fig:0002}. Specifically, we adopt an  unfolding-friendly augmented lagrangian method (ALM) to address the total power constraint, and leverage an unfolded improved davidon-fletcher-powell (DFP) algorithm to overcome the challenges imposed by the non-convex objective function  (\ref{eq:24a}) in the unfolding framework based on iterative ADMM.

\subsection{ADMM-Based Unfolding Framework}
For convenience, we first introduce auxiliary variables $\boldsymbol{V}$ and $\boldsymbol{\Phi}$ as follows:
% \begin{equation}
% \label{eq:25}
% \boldsymbol{V}=[\{\mathbf{v}_k\},\mathbf{v}_t],\quad k\in \mathcal{K},\quad\boldsymbol{\varphi}=[\mathbf{\Theta}_r,\mathbf{\Theta}_t].
% \end{equation}
\begin{equation}
\label{eq:25}
\boldsymbol{V}=[\{\mathbf{v}_k\},\mathbf{v}_t],\quad k\in \mathcal{K},\quad\boldsymbol{\Phi}=[\boldsymbol{\varphi}_r,\boldsymbol{\varphi}_t].
\end{equation}
Likewise, we let
\begin{equation}
\label{eq:26}
\mathbf{W}=[\{\mathbf{w}_k\},\mathbf{w}_t],\quad k\in \mathcal{K}, \quad \mathbf{\Theta}=[\mathbf{\Theta}_r,\mathbf{\Theta}_t].
\end{equation}
Then, we transform the original challenging optimization problem (\ref{24}) into the following optimization problem (\ref{27}):
\begin{subequations}
\label{27}
\begin{equation}
\min_{\mathbf{W},\mathbf{\Theta},\mathbf{V},\boldsymbol{\Phi}}\quad f(\mathbf{W},\mathbf{\Theta}), \label{eq:27a}
\end{equation}
\begin{equation}
\text{s.t.} \hspace{0.5cm}\sum_{k=1}^K||\mathbf{w}_k||^2+||\mathbf{w}_t||^2\leq P_{max}, \label{eq:27b}
\end{equation}
\begin{equation}
|\boldsymbol{\Theta}_r(n)|^2 + |\boldsymbol{\Theta}_t(n)|^2 = 1, \label{eq:27c}
\end{equation}
\begin{equation}
\mathbf{v}_k=\mathbf{w}_k,\,\mathbf{v}_t=\mathbf{w}_t,\,k\in \mathcal{K},
\end{equation}
\begin{equation}
\boldsymbol{\varphi}_r=\boldsymbol{\Theta}_r,\, \boldsymbol{\varphi}_t=\boldsymbol{\Theta}_t,\label{eq:27e}
\end{equation}
\begin{equation}
\sum_{k=1}^K||\mathbf{v}_k||^2+||\mathbf{v}_t||^2\leq P_{max},\label{eq:24f}
\end{equation}
\end{subequations}
where $f(\mathbf{W},\mathbf{\Theta})$ is defined as
\begin{equation}
\label{eq:28}
f(\mathbf{W},\mathbf{\Theta})=\mathbb{E}\left[\frac{1}{K}\sum_{k=1}^K\mathbb{E}\Big[\|\dot{y}_{k,l}-c_{\text{cu},k,l}\|_2^2\Big]\right]+\zeta\cdot[\mathbf{F}^{-1}]_{1,1},
\end{equation}
and then we define
\begin{equation}
\label{eq:29}
q(\mathbf{w}_k,\mathbf{w}_t)=\sum_{k=1}^{K}||\mathbf{w}_k||^2+||\mathbf{w}_t||^2-P_{max}\leq0,
\end{equation}
\begin{equation}
\label{eq:30}
g_{n}(\boldsymbol{\Theta}_{t},\boldsymbol{\Theta}_{r})=|\boldsymbol{\Theta}_{r}(n)|^2+|\boldsymbol{\Theta}_{t}(n)|^2-1=0.
\end{equation}

Subsequently, based on the unfolded framework of ADMM, the problem can be transformed into:
\begin{equation}
\label{eq:31}
\min_{\mathbf{W},\mathbf{\Theta},\mathbf{V},\boldsymbol{\Phi}}\quad L(\mathbf{W},\mathbf{\Theta},\dot{\beta},\dot{\gamma}),
\end{equation}
where $\dot{\beta}$ and $\dot{\gamma}$ are Lagrange multipliers, $L(\mathbf{W},\mathbf{\Theta},\dot{\beta},\dot{\gamma})$ is the Lagrangian of the original problem, and the augmented Lagrangian function can be expressed as follows:
\begin{align}
\small
\label{eq:32}
\begin{split}L(\mathbf{W},\mathbf{\Theta},\dot{\beta},\dot{\gamma})=&f(\mathbf{W},\mathbf{\Theta})+\sum_{n=1}^N\dot{\beta}_ng_n(\boldsymbol{\Theta_t},\boldsymbol{\Theta_r})+\dot{\gamma} q(\mathbf{w}_k,\mathbf{w}_t)\\&+\frac{\dot{\rho}}{2}\|\mathbf{W}-\boldsymbol{V}+\frac{\boldsymbol{\Omega}_1}{\dot{\rho}}\|^2+\frac{\dot{\rho}}{2}\|\boldsymbol{\Theta}-\boldsymbol{\Phi}+\frac{\boldsymbol{\Omega}_2}{\dot{\rho}}\|^2,\end{split}
\end{align}
where \(\dot{\rho} > 0\) is the penalty parameter, and $\boldsymbol{\Omega}_1$  and \(\boldsymbol{\Omega}_2\) are the introduced dual variables. With the aid of the projection-based method, constraints (\ref{eq:27b}) and (\ref{eq:27c}) can always be satisfied in the alternating iterative algorithm.

\subsection{Update $\mathbf{W}$}
In the $t$-th ADMM iteration, given the target problem of $\boldsymbol{V}_t,\mathbf{\Theta}_t,\boldsymbol{\Phi}_t, \boldsymbol{\Omega}_1$ and $\boldsymbol{\Omega}_2$ with respect to $\mathbf{W}$, it can be transformed into an unconstrained problem. By adopting the unfolded ALM with the initial Lagrangian multiplier $\dot{\gamma}_{1}$, the Lagrangian function $L(\mathbf{W},\boldsymbol{\Theta},\dot{\beta}_{j},\dot{\gamma}_j)$ is constructed. For the $j$-th unfolded ALM iteration, $\mathbf{W}_{j+1}$ and $\dot{\gamma}_{j+1}$ can be determined by
\begin{equation}
\label{eq:33}
\mathbf{W}_{j+1}=\arg\min_{\mathbf{W}}L(\mathbf{W}_j,\boldsymbol{\Theta}_j,\dot{\beta}_j,\dot{\gamma}_j),
\end{equation}
% \begin{equation}
% \label{eq:34}
% \dot{\gamma}_{j+1}=\max\{\dot{\gamma}_j+\eta_jq(\mathbf{W}_j,0)\},
% \end{equation}
\begin{equation}
\label{eq:34}
\dot{\gamma}_{j+1}=\max\{\dot{\gamma}_j+\eta_jq_j(\mathbf{w}_k,\mathbf{w}_t),0\},
\end{equation}
where $\eta_{j}$ denotes the trainable dual update parameter. Sequences $\{\dot{\beta}_{j}\}$ and $\{\mathbf{W}_j\}$ are respectively obtained after the ALM iterations. Note that, despite the problem in (\ref{eq:31}) has been simplified, it is still intractable to obtain the optimal $\mathbf{W}$ because of the non-convex subproblem. This challenge can be tackled by introducing the unfolded DFP algorithm. Specifically, the $r$-th iteration of the generalized quasi-Newton method is formulated as follows:
\begin{equation}
\label{eq:35}
\mathbf{W}_{j,r+1}=\mathbf{W}_{j,r}-\tau_{j,r}\mathbf{B}_r\mathcal{G}(\mathbf{W}_{j,r}),
\end{equation}
where $\mathcal{G}(\mathbf{W}_{j,r})=\nabla_{\mathbf{W}}L(\mathbf{W}_{j,r},\boldsymbol{\Theta}_{j},\dot{\beta}_{j},\dot{\gamma}_j)$, and $\tau_{j,r}$ is trainable step-size parameter for gradient-based algorithms \cite{10757550}, and $\mathbf{B}_{r}$ is a symmetric positive definite matrix whose iterative calculation formula is given by
\begin{equation}
\label{eq:36}
\mathbf{B}_{r+1}=\left\{\begin{array}{ll}\mathbf{B}_{r},\quad\quad\qquad\qquad\qquad\qquad\frac{\mathbf{y}_{r}^{T}\delta_{r}}{\|\delta_{r}\|_{2}^{2}}\leq\|\mathcal{G}(\mathbf{W}_{r})\|,\\\mathbf{B}_{r}-\frac{\mathbf{B}_{r}\mathbf{y}_{r}\mathbf{y}_{r}^{T}\mathbf{B}_{r}}{\mathbf{y}_{r}^{T}\mathbf{B}_{r}\mathbf{y}_{r}}+\frac{\mathbf{y}_{r}^{T}\mathbf{B}_{r}\mathbf{y}_{r}}{(\mathbf{y}_{r}^{T}\boldsymbol{\delta}_{r})^{2}}\boldsymbol{\delta}_{r}\boldsymbol{\delta}_{r}^{T},\mathrm{otherwise},\end{array}\right.
\end{equation}
where $\mathbf{y}_r=\mathcal{G}(\mathbf{W}_r)-\mathcal{G}(\mathbf{W}_{r-1})$, $\boldsymbol{\delta}_r=\mathbf{W}_r-\mathbf{W}_{r-1}$. To satisfy the power constraint, we project $\mathbf{W}_{j,r}$ onto the feasible set:
\begin{equation}
\label{eq:37}
\widehat{\mathbf{W}}_{j,r} = 
\begin{cases} 
\mathbf{W}_{j,r}, & \sum_{k=1}^K \|\mathbf{w}_k\|^2 + \|\mathbf{w}_t\|^2 \leq P_{\text{max}}, \\[1ex]
A \cdot [\{\mathbf{w}_k\}, \mathbf{w}_t], & \sum_{k=1}^K \|\mathbf{w}_k\|^2 + \|\mathbf{w}_t\|^2 > P_{\text{max}},
\end{cases}
\end{equation}
where $A=\sqrt{\frac{P_{\max}}{\sum_{k=1}^K\|\mathbf{w}_k\|^2+\|\mathbf{w}_t\|^2}}$.
\subsection{Update $\mathbf{V}$}

Given $\mathbf{W}_{t+1}$, $\boldsymbol{\Theta}_t$, $\boldsymbol{\Phi}_t$, $\boldsymbol{\Omega}_1$, $\boldsymbol{\Omega}_2$, the optimal $\mathbf{V}$ can be achieved by means of
\begin{equation}
\label{eq:38}
\boldsymbol{V}_{t+1}^{\star}=\|\mathbf{W}_{t+1}\|e^{j\angle(\dot{\rho}\mathbf{W}_{t+1} + \boldsymbol{\Omega}_1)},
\end{equation}
where the above achieved $\boldsymbol{V}_{t+1}^{\star}$ is a closed-form optimal solution of $\boldsymbol{V}_{t}$.

\subsection{Update $\mathbf{\Theta}$}
Given the relevant variables, the optimization problem with respect to $\boldsymbol{\Theta}$ is transformed into an unconstrained problem. With the initial Lagrange multiplier $\dot{\beta}_{1,n}$, the Lagrangian function $L(\mathbf{W},\boldsymbol{\Theta},\dot{\beta}_{j,n},\dot{\gamma}_{j,n}), n\in \mathcal{N}=\{1,\ldots,N\}$, is constructed. For the $j$-th ALM iteration, we have\cite{9769997}
\begin{equation}
\label{eq:39}
\boldsymbol{\Theta}_{j+1}=\arg\min_{\boldsymbol{\Theta}}L(\mathbf{W},\boldsymbol{\Theta},\dot{\beta}_{j,n},\dot{\gamma}_{j,n}),
\end{equation}
\begin{equation}
\label{eq:40}
\dot{\beta}_{j+1,n}=\max\{\dot{\beta}_{j,n}+\kappa_{n,j}g_{n,j}(\boldsymbol{\Theta}_{t},\boldsymbol{\Theta}_{r}),0\},
\end{equation}
where $\kappa_{n,j}$ denotes the trainable step-size parameter for dual updates. Similar to problem (\ref{eq:33}), an unfolded DFP algorithm is also utilized to solve the non-convex unconstrained optimization problem (\ref{eq:39}). Specifically, the $r$-th iteration of the generalized quasi-Newton method can be formulated as follows:
\begin{equation}
\label{eq:41}
\mathbf{\Theta}_{j,r+1}=\mathbf{\Theta}_{j,r}-\dot{\zeta}_{j,r}\mathbf{D}_r\mathcal{F}(\mathbf{\Theta}_{j,r}),
\end{equation}
where $\mathcal{F}(\mathbf{\Theta}_{j,r})=\nabla_{\mathbf{\Theta}}L(\mathbf{W},\mathbf{\Theta}_{j,r},\dot{\beta}_{j},\dot{\gamma}_{j})$ is the trainable step-size parameter, and $\dot{\zeta}_{j,r}$ is trainable 
step-size parameter for DFP algorithm. In addition, $\mathbf{D}_{r}$ is a symmetric positive definite matrix computed in a manner similar to $\mathbf{B}_{r}$. 
% \textcolor{red}{$\mathbf{D}_{r}$ represents the approximate Hessian
% matrix obtained in the $r$-th iteration.}
To satisfy the reflection and transmission coefficient constraints, we project $\mathbf{\Theta}_{j,r}$ onto the feasible set:
\begin{equation}
\label{eq:42}
\boldsymbol{\Theta}_{t,j,r}(n)=\sqrt{\beta_t}\frac{\boldsymbol{\Theta}_{t,j,r}(n)}{|\boldsymbol{\Theta}_{t,j,r}(n)|},\forall n=1,\ldots,N,
\end{equation}
\begin{equation}
\label{eq:43}
\boldsymbol{\Theta}_{r,j,r}(n)=\sqrt{\beta_r}\frac{\boldsymbol{\Theta}_{r,j,r}(n)}{|\boldsymbol{\Theta}_{r,j,r}(n)|},\forall n=1,\ldots,N.
\end{equation}

\subsection{Update $\Phi$}
Given $\mathbf{W}_{t+1}$ and $\mathbf{V}_{t+1}$, the closed-form optimal solution for $\boldsymbol{\Phi}$ yields
\begin{equation}
\label{eq:44}
\boldsymbol{\varphi}_{t,t+1}^{\star}=\sqrt{\beta_t}e^{j\angle(\dot{\rho}\boldsymbol{\Theta}_{t,t}+\boldsymbol{\Omega}_2)},
\end{equation}
\begin{equation}
\label{eq:45}
\boldsymbol{\varphi}_{r,t+1}^{\star}=\sqrt{\beta_r}e^{j\angle(\dot{\rho}\boldsymbol{\Theta}_{r,t}+\boldsymbol{\Omega}_2)}.
\end{equation}

\subsection{Update ${\Omega}_1$ and ${\Omega}_2$}
After obtaining $\mathbf{W}_{t+1}$, $\mathbf{V}_{t+1}$, $\boldsymbol{\Theta}_{t+1}$, and $\boldsymbol{\Phi}_{t+1}$, we update the dual variables by means of
\begin{equation}
\label{eq:46}
\boldsymbol{\Omega}_1=\boldsymbol{\Omega}_1+\dot{\rho}(\mathbf{W}_{t+1}-\mathbf{V}_{t+1}),
\end{equation}
\begin{equation}
\label{eq:47}
\boldsymbol{\Omega}_2=\boldsymbol{\Omega}_2+\dot{\rho}(\boldsymbol{\Theta}_{t+1}-\boldsymbol{\Phi}_{t+1}).
\end{equation}

\subsection{Alternating Updates until Convergence}
After the $t$-th ADMM layer iteration, the objective function $f(\mathbf{W},\boldsymbol{\Theta})$ converges if and only if it satisfies the following convergence criterion:
\begin{equation}
\label{eq:48}
|\frac{f_{t+1}(\mathbf{W},\boldsymbol{\Theta})-f_{t}(\mathbf{W},\boldsymbol{\Theta})}{f_{t}(\mathbf{W},\boldsymbol{\Theta})}|\leq\varrho,
\end{equation}
with $\varrho$ denoting the algorithm's convergence threshold. If convergence is not satisfied, we proceed to the ($t$+1)-th ADMM iteration.

\subsection{Training Strategy}
The objective is to obtain the optimal dual step-size update parameters through neural network training. Define the training dataset $\mathcal{H}$ as follows:
\begin{equation}
\label{eq:49}
\mathcal{H}=\{\mathcal{H}_1,\mathcal{H}_2,\ldots,\mathcal{H}_B\},
\end{equation}
where $\mathcal{H}_{b}$ represents the system's CSI dataset for the $b$-th training batch, including channel information such as $\mathbf{H}_{br}$, $\dot{\mathbf{h}}_{rk}$, $\dot{\mathbf{g}}_{rt}$, under channel uncertainties. During the $t$-th ADMM iteration, we define the training parameter set $\mathcal{A}_{t}$, $t=1,\ldots,T$, as
\label{eq:50}
\begin{equation}
\mathcal{A}_t=\{\eta_j,\kappa_j,\tau_{j,r},\dot{\zeta}_{j,r}~|~j=1,\ldots,J,\;r=1,\ldots,R\},
\end{equation}
where $T$, $J$, and $R$
 denote the number of unfolded ADMM layers, unfolded ALM layers, and unfolded DFP layers, respectively, and $B$
 is the number of samples in the training dataset. For each CSI  set, \textit{forward propagation} generates the transmit waveform and STAR-RIS passive beamforming matrix according to the following mapping relationship:
 \begin{equation}
 \label{eq:51}
 \{\mathbf{W}^{}_b,\mathbf{\Theta}^{}_b\}=\mathcal{P}(\mathcal{H}_b,\mathbf{W}_0,\mathbf{\Theta}_0,\mathcal{A}).
 \end{equation}
 During DUN training \cite{10757550}, the loss function denoted by $\mathcal{L}(\mathcal{H}_b)$ is defined as 
  \begin{equation}
  \label{eq:52}
 \mathcal{L}(\mathcal{H}_b):=f(\mathcal{P}(\mathcal{H}_b,\mathbf{W}_0,\mathbf{\Theta}_0,\mathcal{A})).
 \end{equation}

 \textit{Remark 6:} 
 % Notice that the proposed model-based deep learning framework optimizes the simple scalar parameters instead of the trainable matrix parameters that shall be optimized in the existing deep unfolding-based deep networks 
Notice that the proposed model-based deep learning framework optimizes simple scalar parameters, in contrast to the trainable matrix parameters optimized in existing deep unfolding-based networks \cite{9363511,8999801,10117541}.
Thus, the proposed method requires only a small number of training samples, making it eminently suitable for challenging rapidly fluctuating wireless environments. 
In addition, the accurate global channel
information and prior information of STAR-RIS phase shifts are unknown in this paper for improving the feasibility for real-time implementation. 
On the other hand, it is noted that the proposed model-based deep learning framework is not limited to the point-like sensing target scenario. In other words, it is valid to the extended sensing target scenario, by substituting the $\mathrm{CRLB}(\dot{\bar{\mathbf{g}}})$ derived in Lemma \ref{lem:lem03} into the relaxed optimization problem (\ref{27}).
 
%  \subsection{Complex Number Operation Process}
%  %
%  \textcolor{red}{
%  In this paper, the deep unfolding framework is implemented using real-valued functions. Since relevant parameters in the system model, such as elements of the channel matrix and coefficient matrix, are in complex form, the following matrix transformation method is required to convert complex matrix $\mathbf{R}_{}$ into real-valued matrix $\mathbf{I}$. Namely, we have 
% \begin{equation}
% \mathbf{I}_{i,j}=\begin{bmatrix}\Re\{\mathbf{R}_{i,j}\}&-\Im\{\mathbf{R}_{i,j}\}\\\Im\{\mathbf{R}_{i,j}\}&\Re\{\mathbf{R}_{i,j}\}\end{bmatrix}.
%  \end{equation}}

\subsection{Algorithm Analysis}
Note that the proposed model-based deep learning framework for addressing the aforementioned stochastic
optimization problem (\ref{27}) is presented in Algorithm \ref{alg:Framwork2}, which involves joint dual-functional waveform, continuous STAR-RIS reflection, and transmission coefficient optimizations for i) target AoD estimation and ii) distance estimation tasks. In addition, the proposed model-based learning approach in principle is an unfolding of the iterative ADMM algorithm combined with deep learning technique, to reduce the complexity of operations such as large-scale iterations and matrix operations for $\mathbf{W}$ and $\boldsymbol{\Theta}$ updates, while improving the performance trade-off of spectral resource utilization between communication and sensing.

Note also that the computational complexity of the proposed deep learning framework for the STAR-RIS-assisted
ISAC system is incurred by three major operations performed in Algorithm \ref{alg:Framwork2}, i.e., i) the dual-functional waveform updating operation, ii) the STAR-RIS reflection coefficient updating operation, and iii) the deep learning-based ADMM unfolding operation consisting of the forward propagation and back propagation. Specifically, the complexity orders of updating the dual-functional waveform and the STAR-RIS reflection coefficients are
$\mathcal{O}(JR(KM)^{3})$ and $\mathcal{O}(JRN^{3})$, respectively. Besides, the corresponding complexity order of the deep learning-based ADMM unfolding is $\mathcal{O}(TN_{e}(KM+N))$
with $N_e$ being the number of neurons in each layer.
% (with respect to the forward propagation) and  $\mathcal{O}(N_{Par}(KM+N))$ (with respect to the back propagation). 
Therefore, the proposed deep learning framework for estimating target AoD involves an $\mathcal{O}(JR(KM)^{3} + JRN^{3} + TN_{e}(KM+N))$ computational complexity. On the other hand, the memory consumption is incurred by three major operations, i.e., i) the storage for  CSI dataset $\mathcal{H}$, ii) the storage for the dual-functional waveform, STAR-RIS reflection
coefficients, and iii) storage for the deep learning model parameters. We remark that the memory cost is dominated by deep learning model size. As the size of the model (number of layers and neurons) grows, the memory cost increases during the training phase. 
% DUN training adds one-time cost but greatly reduces run-time iterations. 
Our proposed framework unfolds the underlying iterative algorithm to a shallow and lightweight neural network, reducing the computational complexity significantly. Notably,
the average execution time per epoch of our approach will be further evaluated in the following Sec. \ref{sec:simulationire}.

In our system, the BS determines the STAR-RIS configuration and sends control commands via a network-owned controller. As the controlling entity, the BS is typically operated by the network operator and plays a central role in network-controlled modes. In addition, the frequency of phase-shift updates is governed not only by the coherence time but also by the update speed of the channel variation.

\alglanguage{pseudocode}
\begin{algorithm}[!t]
\caption{\small Proposed Model-Based Deep Learning Framework for
AoD Estimation}
\label{alg:Framwork2}
\begin{algorithmic}[1]
\small
\Require
Initialize the matrices $\mathbf{W}_{0}$ and  $\boldsymbol{\Theta}_{0}$, initialize the parameters $\dot{\beta}_{0}$ and $\dot{\gamma}_{0}$, and preset the training hyperparameter set $\mathcal{A}$.
\Ensure Optimized $\mathbf{W}^{\star}$ and  $\boldsymbol{\Theta}^{\star}$.
% \State \multiline{Generate the dataset $\mathcal{H}$ for DUN training and then randomly split $\mathcal{H}$ into $\mathcal{B}$ batches.}
\State {Generate  $\mathcal{H}$ for DUN training and then randomly split $\mathcal{H}$ into $\mathcal{B}$ batches.}
% \State \multiline{Decompose the beamforming matrix $\mathbf{V}$ into low-dimensional components $\mathbf{w}$, $\mathbf{u}$, and $\mu$. }
% $\vartriangleright$ 
\For{{$t$} = 1, $\ldots$, $T$}
\For{$j$ = 1, \ldots, $J$} $\;\vartriangleright$ \textbf{Forward Propagation Phase}
\State {Update the Lagrange multiplier $\dot{\gamma}_{t,j}$ by (\ref{eq:34}). 
 }
\For{$r = 1, \ldots, R$}
%
% \State \multiline{Update the dual-functional waveform $\mathbf{W}_{t,j,r}$  based on gradient-based method following the update criterion (\ref{eq:35}).}
\State {Update  $\mathbf{W}_{t,j,r}$ based on gradient-based method following the update criterion (\ref{eq:35}).}
 \State {Normalize $\mathbf{W}_{t,j,r}$ to satisfy the total power requirement via (\ref{eq:37}).}
 \State {Update $\boldsymbol{\Theta}_{r,j,r}$ and $\boldsymbol{\Theta}_{t,j,r}$ through the generalized quasi-Newton method via (\ref{eq:41}).}
 \State {Normalize $\boldsymbol{\Theta}_{t,j,r}$ and $\boldsymbol{\Theta}_{r,j,r}$ to satisfy the transmission and reflection coefficient constraints by (\ref{eq:42}) and (\ref{eq:43}), respectively.}
 \EndFor
 \EndFor
 \State {Update the auxiliary variable $\mathbf{V}_t$ by (\ref{eq:38}).}
\State {Update the auxiliary variable $\boldsymbol{\Phi}_t$ by (\ref{eq:44}) and (\ref{eq:45}).}
\State {Update the dual variables $\boldsymbol{\Omega}_{1,t}$ and $\boldsymbol{\Omega}_{2,t}$ according to (\ref{eq:46}) and (\ref{eq:47}), respectively.}
%\EndFor
\For{$j = 1, \ldots, J$} $\vartriangleright$ \textbf{Back Propagation Phase}
\State {Update the network parameters $\eta_{t,j}$ and $\kappa_{t,j}$ via the gradient descent method.}
\EndFor
\EndFor
\end{algorithmic}
\end{algorithm}
%-------------------------------->
% %-------------------------------->

%\subsection{Computational Complexity Analysis}
% 
% \textcolor{red}{It is noted that this matrix transformation can convert the complex element at the $i$-th row and $j$-th column of complex matrix $\mathbf{R}$ into a real matrix through the following operations. Let operation $\mathbf{G}(\mathbf{A})$
% denote the process of converting each complex element of the complex matrix $\mathbf{A}\in\mathbb{C}^{N\times M}$ into a real matrix one-to-one, resulting in the real matrix $\mathbf{~B}\in\mathbb{R}^{2N\times2M}$. The corresponding elements can then be extracted and recombined into the complex matrix form, e.g., 
% \begin{align}
% \mathbf{G}(\dot{\mathbf{h}}_{rk}^{H}\mathbf{\Theta}_r\mathbf{H}_{br})=\mathbf{G}(\dot{\mathbf{h}}_{rk}^{H})\mathbf{G}(\mathbf{\Theta}_r)\mathbf{G}(\mathbf{H}_{br}).
% \end{align}
% All complex matrix operations involved in the optimization process are converted into real matrix operations using this transformation.}

% \textcolor{red}{\textit{Remark 5:} Notice that the proposed model-based deep learning framework is not limited to the single sensing target scenario. An extension of the proposed STAR-RIS and NOMA-assisted ISAC transmission system to the case of multiple sensing targets will be our future research topic.}

\section{Simulation Results}
\label{sec:simulationire}
\subsection{Simulation Specification}
In this section, we present simulation results for the proposed STAR-RIS-assisted ISAC system under various BS and STAR-RIS configurations. The simulations validate the effectiveness of the proposed DL-ADMM framework. In the experiments, a simplified model with a single CU and a single SU is considered (although the proposed scheme can be extended to multi-target scenarios). The system model is set in a two-dimensional plane, with the coordinates of the BS, STAR-RIS, CU, and SU as (0, 0) m, (30, 0) m, (35, 7) m, and (25, -7) m, respectively. In addition, we set the wavelength $\lambda=10^{-2}$ m, $\sigma_k^2=\sigma_t^2=-80$ dBm. Unless otherwise specified, the number of STAR-RIS elements is set to be $N$ = 36, the number of BS antennas is set to be $M$ = 16, $L$ = 128, and imperfect CSI is used with a normalized channel mean MSE of $\delta = 0.1$. The communication-sensing trade-off coefficient is $\zeta = 0.5$. During algorithm training, the numbers of unfolded layers are set to be $T$ = 75, $J$ = 8, $R$ = 10 and the training dataset contains $B$ = 500 samples. 
In addition, we use PyTorch 2.1.1
with an NVIDIA RTX 3090 GPU to conduct the following experiments.

For demonstrating the performance benefits and the effectiveness of our proposed framework, we compare our proposed framework with following counterpart approaches:
\begin{itemize}
    \item  \textbf{ADMM-MM} \cite{9769997}: This approach proposes to solve the complicated non-convex optimization problem based on the ADMM and
 majorization-minimization (MM) methods. More concretely, it alternatively optimizes the transmit waveform, RIS reflecting coefficients,  and auxiliary variables until reaching a convergence;  
      \item \textbf{DL-ADMM with RIS (DL-ADMM (RIS))}: This scheme designs the dual-function waveform and RIS reflecting coefficients based on the proposed model-based deep learning framework;
    \item  \textbf{DL-ADMM without RIS (DL-ADMM (w/o RIS))}: This method only designs the dual-function waveform  based on our proposed framework;
    \item \textbf{DL-ADMM with OMA (DL-ADMM (OMA))}: This scheme deploys the STAR-RIS and implements the orthogonal multiple access to avoid co-carrier interference.
     \item \textbf{Unsupervised deep learning with RIS (Unsupervised learning (RIS))}: It jointly optimizes  transmit waveform and RIS reflecting coefficients by employing the objective function loss and constraint-based regular terms.
\end{itemize}

\begin{figure}
\setlength{\abovecaptionskip}{-0.2cm}
\setlength{\belowcaptionskip}{-0.5cm}
    \begin{center}    \includegraphics[width=9cm,height=6.5cm]{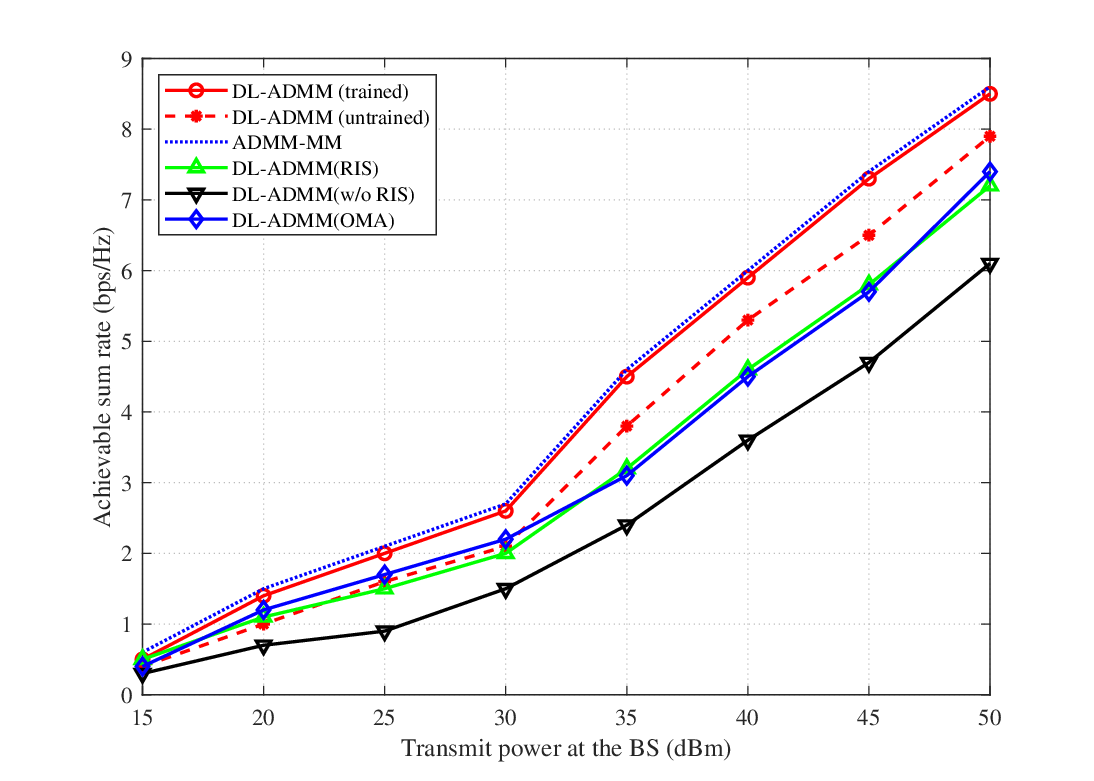}
    \end{center}
    \caption{Achievable sum-rate performance comparisons of distinct resource allocation approaches with different transmit power of BS.}
    \label{fig:002}
\end{figure}

\begin{figure}[tp]
\setlength{\abovecaptionskip}{-0.2cm}
\setlength{\belowcaptionskip}{-0.5cm}
    \begin{center}
    \includegraphics[width=9cm,height=6.5cm]{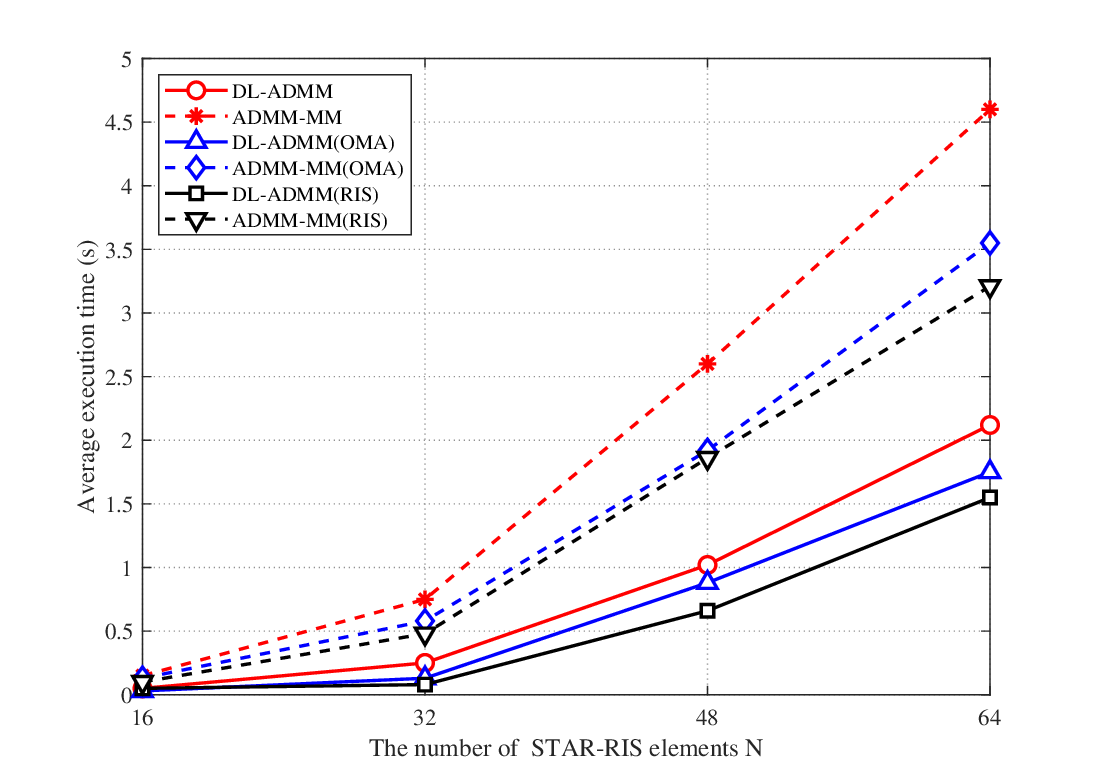}
    \end{center}
    \caption{Average execution time of distinct resource allocation algorithms under different STAR-RIS element configurations.} \label{fig:003}
\end{figure}
\begin{figure}[tp]
\setlength{\abovecaptionskip}{-0.2cm}
\setlength{\belowcaptionskip}{-0.5cm}
    \begin{center}
    \includegraphics[width=9cm,height=6.5cm]{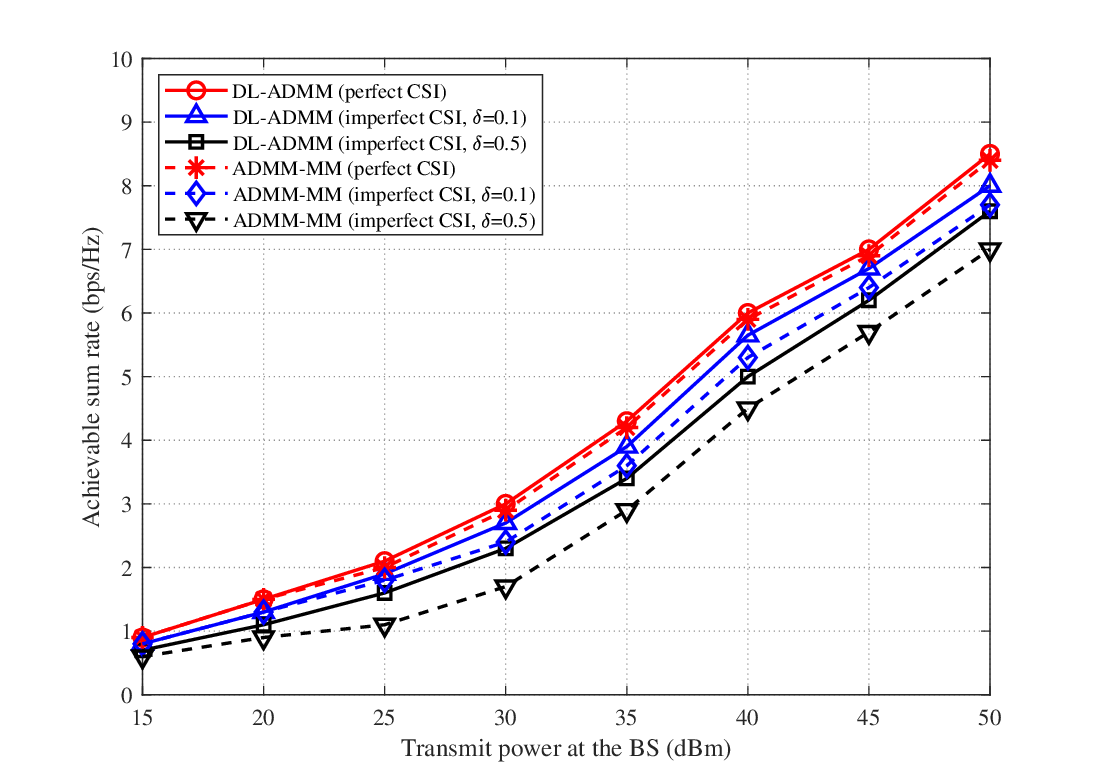}
    \end{center}
    \caption{ Achievable sum rate performance of the proposed ISAC system using different resource allocation approaches under  channel uncertainties.} \label{fig:04}
\end{figure}

\subsection{Achievable Sum-Rate Performance Comparisons}
Fig. \ref{fig:002} illustrates the relationship between the BS transmit power and the system achievable sum rate. The simulation results in Fig. \ref{fig:002} demonstrate that the DL-ADMM algorithm proposed in this paper exhibits superior performance in optimizing the STAR-RIS assisted ISAC system and that STAR-RIS can bring higher gains. The proposed algorithm is compared with the ADMM-MM algorithm proposed in \cite{9769997}, as well as with the DL-ADMM algorithm in scenarios with RIS assistance and without RIS assistance. The simulation results show that the proposed algorithm can achieve results comparable to those of the algorithm in  \cite{9769997} after training. The main reason is that the proposed DL-ADMM approach can gradually minimize the MUI energy and enhance the DoFs in waveform designs in ISAC systems, thus leading to improved communication performance.
Furthermore, under this system model, deploying STAR-RIS achieves a higher system achievable sum rate than both using RIS and without using RIS.

In Fig. \ref{fig:003}, we evaluate the algorithmic complexity by comparing the average execution time. It can be observed from this figure that under different STAR-RIS element numbers $N$ and across various system configurations, the proposed DL-ADMM algorithm consistently outperforms the ADMM-MM algorithm in terms of execution time. Specifically, our proposed approach achieves an average execution time that is approximately $40\%$ shorter than that of the ADMM-MM method.
% it is easy to see that the average execution time of DL-ADMM is approximately $40\%$ of that of ADMM-MM. 
Combining this result with Fig. \ref{fig:002}, it is evident that DL-ADMM can achieve performance comparable to the ADMM-MM algorithm while significantly reducing the execution time, thereby demonstrating its effectiveness in lowering computational complexity and accelerating system waveform and beamforming design.
%This is due to the benefits brought by OSD-aided scheme, essentially coming from the minimizing ${\rm Tr}[(\mathbf{H}_{\mathbf{w}} \mathbf{H}_{\mathbf{w}}^{H})^{-1}]$ by (\ref{equ:220105}).

Fig. \ref{fig:04} compares the impact of channel uncertainty on the performance of the DL-ADMM and ADMM-MM algorithms under different BS transmit powers. The simulation results show that under ideal CSI, the ADMM-MM algorithm achieves a slightly higher system sum rate than DL-ADMM. This phenomenon indicates that when CSI is perfectly accurate, the precise channel parameters utilized by ADMM-MM enable its resource allocation strategy to reach a more optimal theoretical equilibrium. However, under non-ideal CSI conditions, as the channel estimation error rate increases, the advantage of DL-ADMM becomes increasingly evident. DL-ADMM actually outperforms ADMM-MM in terms of system sum rate under the same channel estimation error rate. This is because DL-ADMM accounts for channel uncertainty during the training phase, effectively mitigating the negative impact of channel uncertainty on resource allocation decisions. These results demonstrate that DL-ADMM exhibits superior robustness to channel uncertainty compared to ADMM-MM, validating its robustness in practical scenarios.

\begin{figure}[tp]
\setlength{\abovecaptionskip}{-0.2cm}
\setlength{\belowcaptionskip}{-0.5cm}
\begin{center}    
\includegraphics[width=9cm,height=6.5cm]{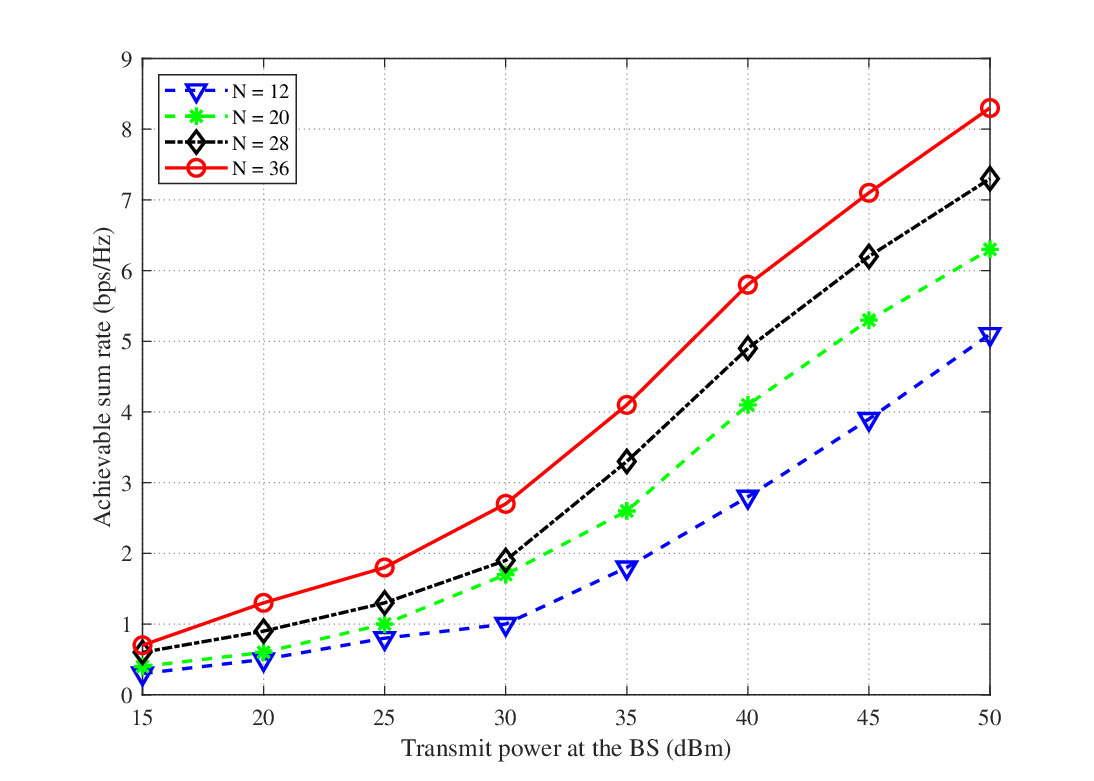}    
\end{center}    
\caption{ Achievable sum rate performance of the proposed ISAC system using DL-ADMM
with different STAR-RIS element configuration $N$.} \label{fig:05}
\end{figure}

\begin{figure}[tp]
\setlength{\abovecaptionskip}{-0.2cm}
\setlength{\belowcaptionskip}{-0.5cm}    
\begin{center}    
\includegraphics[width=9cm,height=6.5cm]{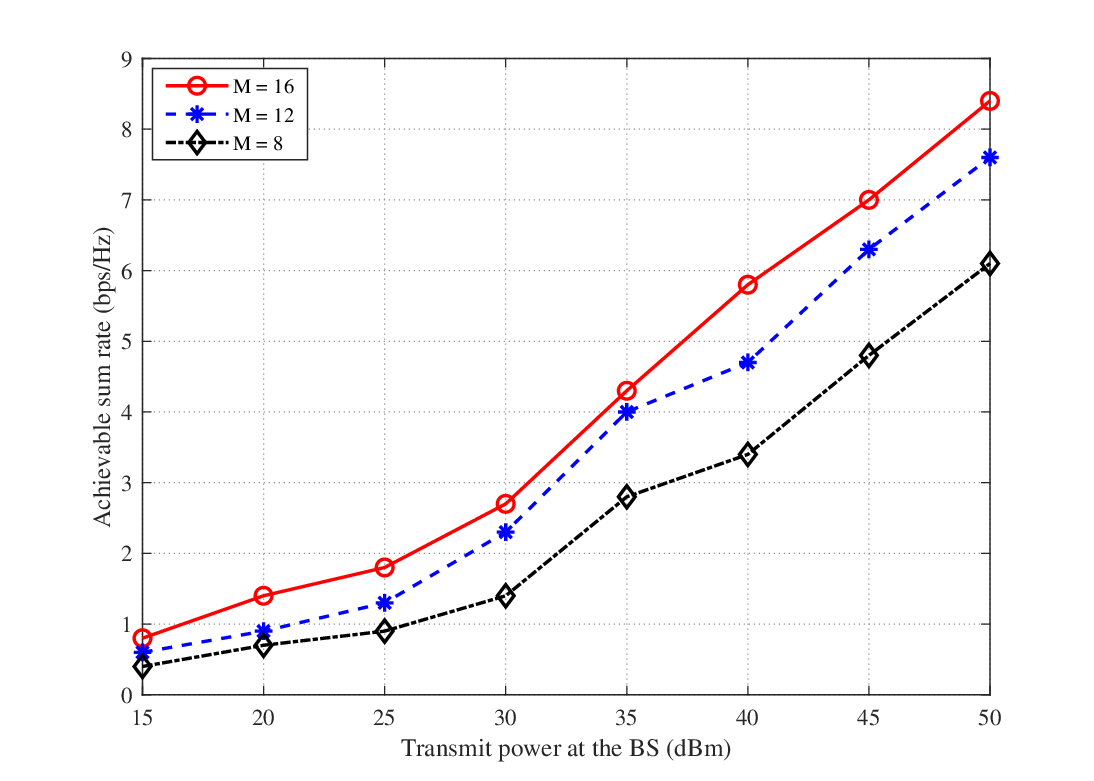}    
\end{center}    
\caption{ Achievable sum rate performance of the proposed ISAC system using DL-ADMM
with different BS antenna configuration $M$.} \label{fig:06}
\end{figure}
\begin{figure}[tp]
\setlength{\abovecaptionskip}{-0.3cm}
\setlength{\belowcaptionskip}{-0.6cm}    
\begin{center}    
\includegraphics[width=9cm,height=6.5cm]{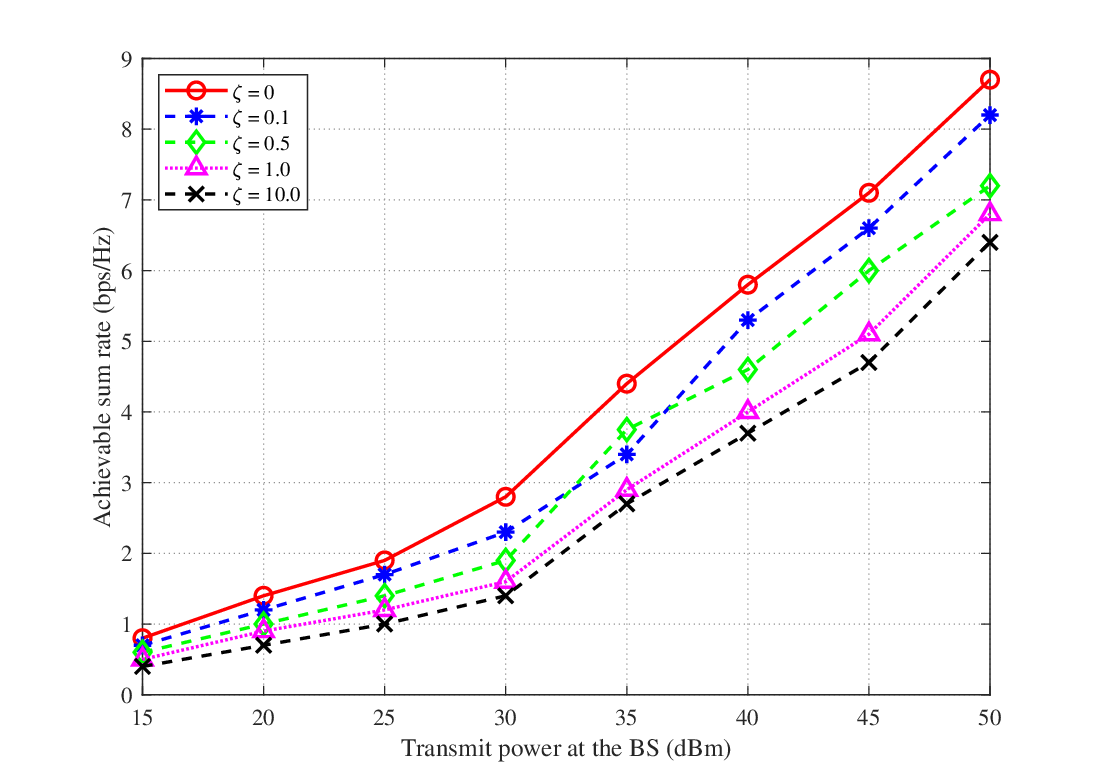}    
\end{center}    
\caption{  Achievable sum rate performance of the proposed ISAC system using DL-ADMM
with different communication-sensing trade-off coefficients.} \label{fig:07}
\end{figure}

In Fig. \ref{fig:05}, we depict the system sum rate performance of the proposed ISAC system using STAR-RIS with different numbers of passive elements $N$. It can be observed from Fig. \ref{fig:05} that a larger number of STAR-RIS elements enables a higher system sum rate. This is because increasing the number of elements allows for more precise phase adjustment and amplitude control of the signals received from the BS, achieving more accurate beamforming and improving the system rate. However, as shown in Fig. \ref{fig:003}, we found out that increasing the number of STAR-RIS elements will prolong the algorithm execution time and increase computational complexity.  

Similarly, Fig. \ref{fig:06} illustrates the achievable system sum rate of the proposed STAR-RIS-assisted ISAC system with different numbers of BS antennas $M$. It can be seen that using more antennas $M$ leads to a higher system sum rate. This is because increasing the number of antennas enhances the system's spatial DoF and enables more precise adjustment of the phase and amplitude of transmitted signals, allowing the signals to be more concentrated towards target users and thus improving the system sum rate.

Fig. \ref{fig:07} compares the achievable system sum rate under different communication-perception trade-off coefficients 
$\zeta$. It can be observed from Fig. \ref{fig:07} that the smaller the trade-off coefficient $\zeta$, the higher the system achievable sum rate. When the trade-off coefficient $\zeta$ is smaller, the optimization places more focus on communication performance metrics.
% \footnote{Note that the Pareto boundary is an important metric for ISAC to identify the optimal balance between communication and sensing performance\cite{10540149}. In this sense, the MMSE Pareto boundary and the CRB Pareto boundary, as well as the waveform design to achieve Pareto optimality in ISAC will be studied in our future research.} 
Notably, the MUI energy is negatively correlated with the signal rate, i.e., reducing MUI energy increases the system sum rate. In addition, as the trade-off coefficient $\zeta$ increases from 1.0 to 10.0, the corresponding performance gap of the two systems in terms of the achievable sum rate diminishes, due to the mutual constraint between communication and perception performance, tending towards a critical value.

\subsection{AoD Estimation Performance Comparisons}
In this subsection, we evaluate the CRLB performance for target AoD estimation task. To demonstrate the AoD estimation performance of the proposed model-based deep learning approach as well as the performance gain brought by STAR-RIS, we compare our approach with the unsupervised learning-based scheme, which exhibits low complexity.   

\begin{figure}[t]
\setlength{\abovecaptionskip}{-0.1cm}
\setlength{\belowcaptionskip}{-0.5cm}    
\begin{center}    
\includegraphics[width=9cm,height=6.5cm]{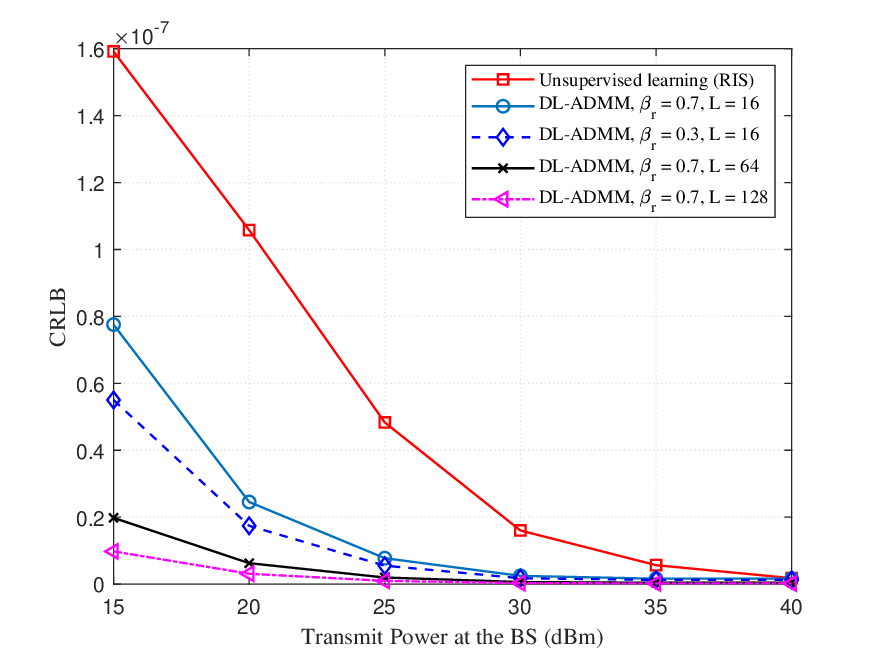}    
\end{center}    
\caption{  AoD estimation performance using distinct resource allocation approaches with $N$ = 36, $\beta_r$ = $\{0.3, 0.7\}$, and $L$ = \{16, 64, 128\}.} \label{fig:09}
\end{figure}

Fig. \ref{fig:09} depicts the CRLB performance of the proposed ISAC system impacted by the transmit power, employing different resource allocation algorithms. In general, the CRLB for estimating AoD decreases when increasing the transmit power. This is consistent with the analytical results given in Lemma \ref{lem02}. This figure also reveals that the proposed approach can achieve 
higher AoD estimation accuracy relative to the unsupervised learning-based approach.
% superior CRLB performance than both the unsupervised learning-based approach and the DL-ADMM with random phase-shifting control. Implicitly, this indicates that our approach enables higher AoD estimation accuracy relative to above baseline methods.
One of the main reasons is that the estimation accuracy is highly sensitive to channel errors and path loss, leading to significant degradation as the CSI error variance increases.
However, our model-based learning framework demonstrates inherent robustness to channel uncertainties by employing the gradient-based training approach that leverages imperfect CSI data during the model training phase.
% our model-based learning framework exhibits inherent robustness against channel uncertainties, as we apply the gradient-based method to train the DUN model using imperfect CSI data during the model training stage. 
% thus rendering our algorithm’s inherent robustness against channel uncertainties. 
Another reason is that our approach optimizes PS factor $\beta_r$ and  
phase-shifting coefficients $\boldsymbol{\Theta}$ of STAR-RIS, thus providing additional DoFs and improving the  estimation accuracy. 
% Besides, this figure indicates that decreasing the PS factor $\beta_r$ improves the AoD estimation performance.
% This demonstrates the effectiveness  of the proposed approach in terms of improving the performance trade-off between communication and sensing at the level of resource allocation in  ISAC, particular with a low PS factor $\beta_r$. 
As expected, this figure indicates that increasing the amount of measured digital samples (i.e., time-slots $L$) sharply improves the angular estimation accuracy from the measurement. The major reason is that increasing $L$ can increase the rank of the signal covariance matrix $\mathbf{R}_{{\Lambda}}$, thus lowering the trace of $\mathbf{R}_{{\Lambda}}^{-1}$, which in principle represents the amount of information of the signal covariance matrix, and yielding lower CRLB.

Fig. \ref{fig:10} illustrates the CRLB performance of the proposed ISAC system impacted by the PS factor $\beta_r$ and the communication-sensing trade-off coefficient $\zeta$. 
% It is easy to see from this figure that the AoD estimation performance is sensitive to the PS factor $\beta_r$. With given transmit power, decreasing the value of $\beta_r$ from 0.7 to 0.3 leading to significant estimation performance improvement. 
It is easy to see from this figure that increasing the value of trade-off coefficient $\zeta$ improves the estimation performance, especially for a high $\zeta$ scenario. The main reasons are that: i) a high $\zeta$ can turn a marginally stable unfolded network into an exploding gradient network, especially when encountering a very deep unfolded network with many nested layers; ii)  high $\zeta$ can
bring more focus on the sensing performance metric during the optimization process.
% ; ii) Increasing  $\zeta$ can  increase the effective Lipschitz constant of the gradient flow when using bounded step sizes. 
% However, this will cause unstable forward iterates or exploding backpropagated gradients unless step sizes are correspondingly bounded.

\begin{figure}[t]
\setlength{\abovecaptionskip}{-0.1cm}
\setlength{\belowcaptionskip}{-0.5cm}    
\begin{center}    
\includegraphics[width=9cm,height=6.5cm]{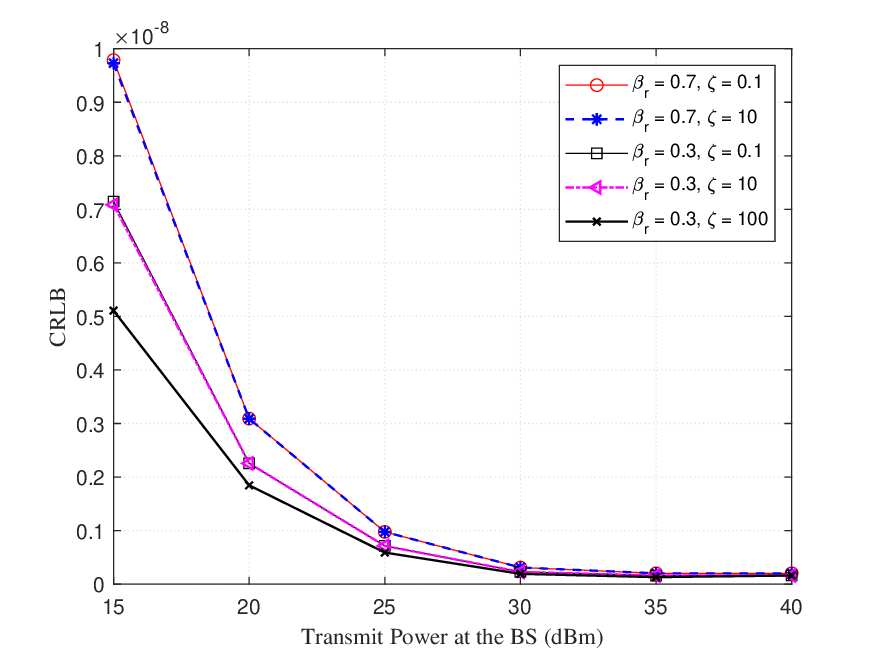}    
\end{center}    
\caption{  AoD estimation performance of the proposed ISAC system using DL-ADMM
with different combinations of the PS factor $\beta_r$ and the communication-sensing trade-off coefficient $\zeta$.} \label{fig:10}
\end{figure}

\section{Conclusion and Future Research Directions}
\label{sec:Con}
This paper investigated a STAR-RIS-assisted ISAC system, focusing on the joint optimization of BS transmit beamforming vectors and STAR-RIS coefficient matrices. An optimization problem model was constructed to minimize the weighted sum of communication and sensing performance metrics. To reduce computational complexity, a lightweight deep learning framework named DL-ADMM was proposed. Simulation results demonstrated that DL-ADMM achieves performance comparable to the ADMM-MM algorithm while significantly reducing execution time, enabling rapid beamforming optimization.
Considering the challenge of acquiring accurate CSI in practical applications, this study introduces additive error interference to simulate imperfect CSI. Experimental results verified that the proposed algorithm exhibits robust performance under channel uncertainties.

Based on this work, there are still some potential research directions that deserve further exploration.
\begin{itemize}
%     \item An extension of our proposed STAR-RIS-assisted ISAC frame-
% work adopting other communication performance metrics, such as spectral
% efficiency [9], energy efficiency and practical QoS, or user fairness [1], [36],
% will be our future research topic.
%
\item 
% Note that the potential deployment scenes of the STAR-RIS controller
% include the ultra-high dense urban environment, remote or coverage blind
% areas. These scenarios have a high demand for low-latency 6G ISAC commu-
% nication. 
In our next step research,  we will investigate task-oriented
resource allocation for ISAC systems assisted by STAR-RIS, with a focus
on vehicular network applications under stringent latency constraints for
autonomous driving.
\item  Note that the Pareto boundary is an important metric for ISAC to identify the optimal balance between communication and sensing performance\cite{10540149}. In this sense, the MMSE Pareto boundary and the CRB Pareto boundary, as well as the waveform design to achieve Pareto optimality in ISAC will be studied in our future research.
\item  
% Note that the current trainable step sizes ensure positivity, but not Lipschitz stability. 
For a very deep unfolded network with many nested layers, using bounded step sizes and objective normalization may ensure the stability of the gradient during back propagation. 
In addition, the communication-sensing weight can increase the effective Lipschitz constant of the gradient flow when the step sizes are correspondingly bounded. 
We will explore the impacts of the trainable step sizes on the Lipschitz stability of our framework and improve the proposed framework for very large deep unfolded network scenario in our future study.
\end{itemize}

\begin{appendices}
%----------------------> Lemma 1
\section{Proof of Lemma 1}
\label{sec:Proof1}
Recall that the unknown channel parameters to be estimated are those contained in $\dot{\mathbf{g}}_{rt}$, i.e., the 2-tuple vector $\boldsymbol{\nu}:=[\theta_2,d_2]^{T}$. Since the additive noise is complex Gaussian distributed, by introducing noise-free $\boldsymbol{\mu}(\boldsymbol{\nu}):=[\dot{\mathbf{g}}_{rt}^H\mathbf{\Theta}_t\mathbf{H}_{br}\mathbf{w}_1s_{1,1}+\dot{\mathbf{g}}_{rt}^H\mathbf{\Theta}_t\mathbf{H}_{br}\mathbf{w}_ts_{t,1},\ldots,\dot{\mathbf{g}}_{rt}^H\mathbf{\Theta}_t\mathbf{H}_{br}\mathbf{w}_1s_{1,L}+\dot{\mathbf{g}}_{rt}^H\mathbf{\Theta}_t\mathbf{H}_{br}\mathbf{w}_ts_{t,L}]^{T} = (\dot{\mathbf{g}}_{rt}^H\boldsymbol{\Theta}_t\mathbf{H}_{br}\ddot{\mathbf{W}}\dot{\mathbf{S}})^{T}$, the FIM for $\boldsymbol{\nu}$ is given by
$
\mathbf{F}_{ij}:=\frac {2}{\sigma_t^2}\Re\{\frac{\partial\boldsymbol{\mu}^{\mathrm{H}}}{\partial\boldsymbol{\nu}_i}\frac{\partial\boldsymbol{\mu}}{\partial\boldsymbol{\nu}_j}\},i,j=1,2.$ Herein, $\ddot{\mathbf{W}} = \left[\mathbf{w}_1,\mathbf{w}_t\right]$,
$\dot{\mathbf{S}} = \left[\dot{\mathbf{s}}_{_1},\ldots, \dot{\mathbf{s}}_{_{L}}\right]$, and $\dot{\mathbf{s}}_{_l}=[s_{1,l}, s_{t,l}]^{T}$.
As such, the FIM matrix ${\mathbf{F}}$ can be written as
\begin{equation}
\label{eq:0060}
\small
\mathbf{F}=\begin{bmatrix}{F}_{\theta\theta}&{F}_{\theta d}\\{F}_{\theta d}&{{F}}_{dd}\end{bmatrix},\end{equation}
with ${F}_{\theta\theta}$, ${{F}}_{\theta d}$, ${{F}}_{dd}$ being
\begin{align}
\label{eq:0061}
\small
\begin{split}
{F}_{\theta\theta}&=\frac{2}{\sigma_t^{2}}\Re\{\big(\big(\dot{\mathbf{h}} \ddot{\mathbf{W}}
\dot{\mathbf{S}}\big)^{T}\big)^{H}  \cdot  \big(\dot{\mathbf{h}} \ddot{\mathbf{W}}
\dot{\mathbf{S}}\big)^{T}\}\\
&=\frac{2}{\sigma_t^{2}} \mathrm{Tr}\Big(\dot{\mathbf{S}}^{T} \ddot{\mathbf{W}}^{T} \dot{\mathbf{h}}^{T} \cdot \dot{\mathbf{h}}^{*} \ddot{\mathbf{W}}^{*}
\dot{\mathbf{S}}^{*}\Big)\\
&=\frac{2}{\sigma_t^{2}} \mathrm{Tr}\Big(\mathbf{R}_x[\mathbf{w}_1,\mathbf{w}_t][\mathbf{w}_1,\mathbf{w}_t]^{H}\Big),
\end{split}
\end{align}
\begin{align}
\label{eq:0062}
\small
\begin{split}
{{F}}_{\theta d} &=\frac{2}{\sigma_t^{2}}\Re\{((\dot{\mathbf{h}} \dot{\mathbf{W}}
\mathbf{S})^{T})^{H} \cdot (\mathbf{h} \dot{\mathbf{W}}
\mathbf{S}D)^{T}\}\\
& = -\frac{2}{\sigma_t^{2}d_{rt}} \mathrm{Tr}\Big(\dot{\mathbf{S}}^{T} \ddot{\mathbf{W}}^{T} {\mathbf{h}}^{T} \cdot \dot{\mathbf{h}}^{*} \ddot{\mathbf{W}}^{*}\dot{
\mathbf{S}}^{*}\Big)\\
& = -\frac{2}{\sigma_t^{2}d_{rt}} \mathrm{Tr}\Big(\mathbf{R}_z[\mathbf{w}_1,\mathbf{w}_t][\mathbf{w}_1,\mathbf{w}_t]^{H}\Big),
\end{split}
\end{align}
\begin{align}
\label{eq:0063}
\small
\begin{split}
{{F}}_{dd}=&\frac{2}{\sigma_t^{2}}\Re\{((\mathbf{h} \ddot{\mathbf{W}}
\dot{\mathbf{S}}D)^{T})^{H}  \cdot  (\mathbf{h} \ddot{\mathbf{W}}
\dot{\mathbf{S}}D)^{T}\}\\
=&\frac{2}{\sigma_t^{2}d_{rt}^2}\Big(1+\frac{4\pi^2d_{rt}^2}{\lambda^2}\Big) \mathrm{Tr}\Big(\dot{\mathbf{S}}^{T} \ddot{\mathbf{W}}^{T} \mathbf{h}^{T} \cdot \mathbf{h}^{*} \ddot{\mathbf{W}}^{*}
\dot{\mathbf{S}}^{*}\Big)\\
=&\frac{2}{\sigma_t^{2}d_{rt}^2}\Big(1+\frac{4\pi^2d_{rt}^2}{\lambda^2}\Big) \mathrm{Tr}\Big(\mathbf{R}_y [\mathbf{w}_1,\mathbf{w}_t][\mathbf{w}_1,\mathbf{w}_t]^{H} \Big),
\end{split}
\end{align}
in which $\dot{\mathbf{h}}=\dot{\mathbf{g}}_{rt}^{H}\mathbf{\Theta}_{t}\mathbf{D}_{rt}^{H}\mathbf{H}_{br} \in\mathbb{C}^{1\times M}$, and $\mathbf{h} = \dot{\mathbf{g}}_{rt}^{H}\mathbf{\Theta}_{t}\mathbf{H}_{br} \in\mathbb{C}^{1\times M}$. 
% In addition,
% $D = \frac{-\frac{j2\pi d_{rt}}{\lambda}-1}{d_{rt}}$, $\mathbf{S}_1 = \text{diag}(\mathbf{1}\otimes {\mathbf{s}})$, $\mathbf{s}=[s_1,\ldots,s_K,s_t]^{T}$, and $\boldsymbol{\omega}_1 = \text{diag}(\mathbf{w}_1,\ldots,\mathbf{w}_K,\mathbf{w}_t)$. Besides, $\mathbf{1}$ corresponds to the ($K$+1)-element all-one vector.

Next, we derive the CRLB for estimating the target AoD $\theta_{rt}$, which corresponds to the first diagonal element of ${\mathbf{F}}^{-1}$:
\begin{align}
\label{eq:00065}
\small
{\mathrm{CRLB}}(\theta_{})=[{\mathbf{F}}^{-1}]_{1,1}=({{F}}_{\theta\theta}-{{F}}_{\theta d}{{F}}^{-1}_{dd}{{F}}_{\theta d})^{-1}.
% &= \frac{\sigma_{t}^2}{2}\left(\textrm{Tr}(\mathbf{R}_x \boldsymbol{\Omega}) 
% - \frac{|\textrm{Tr}(\mathbf{R}_{y} \boldsymbol{\Omega})\Re\{D\}|^{2}}{\textrm{Tr}(\mathbf{R}_{z} \boldsymbol{\Omega})D\cdot D^{*}}
% \right)^{-1},
\end{align}
Then, substituting the Eqs. (\ref{eq:0061}), (\ref{eq:0062}), and (\ref{eq:0063}) into the above (\ref{eq:00065}) yields 
% \begin{align}
% \small
% \begin{split}
% &{\mathrm{CRLB}}(\theta_{rt})=
% % [\underline{\mathbf{F}}^{-1}]_{1,1}=[\underline{\mathbf{F}}_{\theta\theta}-\underline{\mathbf{F}}_{\theta d}\underline{\mathbf{F}}^{-1}_{dd}\underline{\mathbf{F}}_{\theta d}^T]^{-1}\\
% \frac{\sigma_{t}^2}{2}\left(\textrm{Tr}(\mathbf{R}_x \boldsymbol{\Omega}) 
% - \frac{|\textrm{Tr}(\mathbf{R}_{y} \boldsymbol{\Omega})\Re\{D\}|^{2}}{\textrm{Tr}(\mathbf{R}_{z} \boldsymbol{\Omega})D\cdot D^{*}}
% \right)^{-1}\\
% & =\frac{\sigma_t^2}{2}\frac{1}{\mathrm{Tr}((\dot{\mathbf{H}}\mathbf{W})^H(\dot{\mathbf{H}}\mathbf{W}))-\frac{(\mathrm{Tr}((\dot{\mathbf{H}}\mathbf{W})^H(\mathbf{H}\mathbf{W}D)))^2}{\mathrm{Tr}((\mathbf{H}\mathbf{W}D)^H(\mathbf{H}\mathbf{W}D))}}.
% \end{split}
% \end{align}
% \begin{align}
% \begin{split}
% \small
% &{\mathrm{CRLB}}(\theta_{rt})=\\&\frac{\frac{\sigma_t^{2}}{2}}{\mathrm{Tr}((\dot{\mathbf{H}}\boldsymbol{\omega}\mathbf{S})^{H}(\dot{\mathbf{H}}\boldsymbol{\omega}\mathbf{S}))-\frac{\lambda^2|\mathrm{Tr}((\dot{\mathbf{H}}\boldsymbol{\omega}\mathbf{S})^{H}(\mathbf{H}\boldsymbol{\omega}\mathbf{S}))|^{2}}{(\lambda^2+4\pi^2d_{rt}^2)\mathrm{Tr}((\mathbf{H}\boldsymbol{\omega}\mathbf{S})^{H}(\mathbf{H}\boldsymbol{\omega}\mathbf{S}))}}.
% \end{split}
% \end{align}
\begin{align}
\begin{split}
\small
&{\mathrm{CRLB}}(\theta_{})=\\&\frac{\frac{\sigma_t^{2}}{2}}{\mathrm{Tr}((\dot{\mathbf{h}}\ddot{\mathbf{W}}\dot{\mathbf{S}})^{T} (\dot{\mathbf{h}}^{} \ddot{\mathbf{W}}^{}
\dot{\mathbf{S}}^{})^{*})-\frac{\lambda^2|\mathrm{Tr}(({\mathbf{h}}\ddot{\mathbf{W}}\dot{\mathbf{S}})^{T} (\dot{\mathbf{h}}^{} \ddot{\mathbf{W}}^{}
\dot{\mathbf{S}}^{})^{*})|^{2}}{(\lambda^2+4\pi^2d_{rt}^2)\mathrm{Tr}(({\mathbf{h}}\ddot{\mathbf{W}}\dot{\mathbf{S}})^{T} ({\mathbf{h}}^{} \ddot{\mathbf{W}}^{}
\dot{\mathbf{S}}^{})^{*})}},
\end{split}
\end{align}
which completes the proof.
%\end{appendices}

%\begin{appendices}
%----------------------> Lemma 2
\section{Proof of Lemma 2} 
\label{sec:proof2}
When taking the extended target case into accounts, the corresponding unknown channel parameters to
be estimated can be denoted by 
\begin{equation}
\label{eq:00103}
    \dot{\mathbf{G}}_{rt} = \sum_{i=1}^{N_s}
    \frac{e^{-j2\pi d_{i}/\lambda}}{\sqrt{\rho_{i}}} \boldsymbol{\alpha}_{2}(\theta_{i})\boldsymbol{\alpha}_3^{H}(\theta_i) +e_{rt},
\end{equation}
where $N_s$ is the number of scatterers. $d_i$ represents the distance between the STAR-RIS and the $i$th scatterer of the target. In addition, $\rho_i$ is the path loss, i.e., $\rho_i = d_i^2$. Besides, 
$\boldsymbol{\alpha}_3(\theta_i) \in \mathbb{C}^{T\times 1}$ accounts for the array response vector of the receive antennas. 

Since the additive noise follows complex Gaussian distribution, we next introduce the noise-free sensing signal 
% \begin{align}
% \small
% \label{eq:0069}
%     \begin{aligned}
% \boldsymbol{\mu}(\dot{\mathbf{G}}_{rt})  =&\Big[\sum_{k=1}^{K}\dot{\mathbf{G}}^{H}_{rt}\boldsymbol{\Theta}_{t}\mathbf{H}_{br}\mathbf{w}_{k}s_{k,1}+\dot{\mathbf{G}}^{H}_{rt}\boldsymbol{\Theta}_{t}\mathbf{H}_{br}\mathbf{w}_{t}s_{t,1}, \ldots, \\&\sum_{k=1}^{K}\dot{\mathbf{G}}^{H}_{rt}\boldsymbol{\Theta}_{t}\mathbf{H}_{br}\mathbf{w}_{k}s_{k,L}+\dot{\mathbf{G}}^{H}_{rt}\boldsymbol{\Theta}_{t}\mathbf{H}_{br}\mathbf{w}_{t}s_{t,L}\Big]^{T} \\
% =&\Big(\dot{\mathbf{G}}^{H}_{rt}\boldsymbol{\Theta}_{t}\mathbf{H}_{br}\dot{\mathbf{W}}\mathbf{S}\Big)^{T},
%  % =\sum_{k=1}^{K}\dot{\mathbf{G}}^{H}_{rt}\boldsymbol{\upsilon}_{k}+\dot{\mathbf{G}}^{H}_{rt}\boldsymbol{\upsilon}_t,
% \end{aligned}
% \end{align}
\begin{align}
\small
\label{eq:0069}
    \begin{aligned}
\boldsymbol{\mu}(\dot{\mathbf{G}}_{rt})  =&\Big[\sum_{k=1}^{K}\dot{\mathbf{G}}^{H}_{rt}\boldsymbol{\Theta}_{t}\mathbf{H}_{br}\mathbf{w}_{k}s_{k,1}+\dot{\mathbf{G}}^{H}_{rt}\boldsymbol{\Theta}_{t}\mathbf{H}_{br}\mathbf{w}_{t}s_{t,1}, \ldots, \\&\sum_{k=1}^{K}\dot{\mathbf{G}}^{H}_{rt}\boldsymbol{\Theta}_{t}\mathbf{H}_{br}\mathbf{w}_{k}s_{k,L}+\dot{\mathbf{G}}^{H}_{rt}\boldsymbol{\Theta}_{t}\mathbf{H}_{br}\mathbf{w}_{t}s_{t,L}\Big]^{T} \\
=&\Big(\dot{\mathbf{G}}^{H}_{rt}\boldsymbol{\Lambda}\Big)^{T},
 % =\sum_{k=1}^{K}\dot{\mathbf{G}}^{H}_{rt}\boldsymbol{\upsilon}_{k}+\dot{\mathbf{G}}^{H}_{rt}\boldsymbol{\upsilon}_t,
\end{aligned}
\end{align}
where $\boldsymbol{\Lambda}_{}=\boldsymbol{\Theta}_{t}\mathbf{H}_{br}\dot{\mathbf{W}}\mathbf{S}$. In this case, (\ref{eq:0069}) is a linear noise-free model in $\dot{\mathbf{G}}_{rt}$. By vectorizing the $\dot{\mathbf{G}}^{H}_{rt}\boldsymbol{\Lambda}$ item in (\ref{eq:0069}), we have
% \begin{align}
% \begin{split}
% \mathrm{vec}(\dot{\mathbf{G}}^{H}_{rt}\boldsymbol{\upsilon}_{k})&=(\boldsymbol{\upsilon}_{k}^{T}\otimes \mathbf{I}_{L})\mathrm{vec}(\dot{\mathbf{G}}^{H}_{rt}) \\&=(s_k\mathbf{w}_{k}^{T}\mathbf{H}_{br}^{T}\boldsymbol{\Theta}_{t}^{T}\otimes \mathbf{I}_{L})\mathrm{vec}(\dot{\mathbf{G}}^{H}_{rt}),
% \end{split}
% \end{align}
\begin{align}
\label{eq:067}
\begin{split}
\mathrm{vec}(\dot{\mathbf{G}}^{H}_{rt}\boldsymbol{\Lambda})&\overset{(\text{I})}{=} (\mathbf{I}_{T} \otimes \boldsymbol{\Lambda}_{}^{T})\mathrm{vec}((\dot{\mathbf{G}}^{H}_{rt})^T) \\
& \overset{(\text{II})}{=}  (\mathbf{I}_{T} \otimes \boldsymbol{\Lambda}_{}^{H})\mathrm{vec}(\dot{\mathbf{G}}^{}_{rt}) 
\\&=(\mathbf{I}_{T} \otimes \mathbf{S}^H\dot{\mathbf{W}}^{H}\mathbf{H}_{br}^{H}\boldsymbol{\Theta}_{t}^{H})\mathrm{vec}(\dot{\mathbf{G}}^{}_{rt}),
\end{split}
\end{align}
% \begin{align}
% \begin{split}
% \mathrm{vec}(\dot{\mathbf{G}}^{H}_{rt}\boldsymbol{\upsilon}_{t})&=(\boldsymbol{\upsilon}_{t}^{T}\otimes \mathbf{I}_{L})\mathrm{vec}(\dot{\mathbf{G}}^{H}_{rt}) \\&=(s_t\mathbf{w}_{t}^{T}\mathbf{H}_{br}^{T}\boldsymbol{\Theta}_{t}^{T}\otimes \mathbf{I}_{L})\mathrm{vec}(\dot{\mathbf{G}}^{H}_{rt}),
% \end{split}
% \end{align}
where $\mathbf{I}_T$ is an identity matrix of size $T\times T$. Notably, the step (I) in (\ref{eq:067}) is based on the property that $\mathrm{vec}\left(\mathbf{ABC}\right)=\left(\mathbf{C}^{T}\otimes\mathbf{A}\right)\mathrm{vec}\left(\mathbf{B}\right)$ \cite{11145277}. The step (II) in (\ref{eq:067}) is based on the property that $\left(\mathbf{I}^{}\otimes\mathbf{B}^{T}\right)\mathrm{vec}\left(\mathbf{C}^{*}\right) = \left(\mathbf{I}^{}\otimes\mathbf{B}^{H}\right)\mathrm{vec}\left(\mathbf{C}^{}\right)$.
Accordingly, the above formula (\ref{eq:0069}) can be vectorized in the form of
\begin{align}
\label{eq:0072}
\begin{split}
\boldsymbol{\mu}(\dot{\mathbf{G}}_{rt}) =(\mathbf{I}_{T} \otimes \boldsymbol{\Lambda}_{}^{H})\mathrm{vec}(\dot{\mathbf{G}}^{}_{rt}).
\end{split}
\end{align}

Subsequently, the derivative of $\boldsymbol{\mu}(\dot{\mathbf{G}}_{rt}) $ in (\ref{eq:0072}) with respect to the $\dot{\bar{\mathbf{g}}}_{} = \text{vec}(\dot{\mathbf{G}}_{rt}^{})$ is given by \begin{align}
    % \frac{\partial \boldsymbol{\mu}}{\partial \bar{\mathbf{g}}_{rt}}=\sum_{k=1}^{K}\boldsymbol{\upsilon}_{k}^{T}\otimes \mathbf{I}_{L}+ \boldsymbol{\upsilon}_{t}^{T}\otimes \mathbf{I}_{L}.
    \frac{\partial \boldsymbol{\mu}}{\partial \dot{\bar{\mathbf{g}}}_{}}=\mathbf{I}_{T} \otimes \boldsymbol{\Lambda}_{}^{H}.
\end{align}
% where $\mathbf{A}$ is the commutable matrix satisfying that $\mathbf{A} \mathrm{vec}(\dot{\mathbf{G}}^{H}_{rt}) = \mathrm{vec}((\dot{\mathbf{G}}^{H}_{rt})^T)$. 
Hence, the corresponding FIM with respect to $\dot{\bar{\mathbf{g}}}_{}$ is known to be \cite{9364358}
% \begin{align}
%     \mathbf{F}(\bar{\mathbf{g}}_{rt} )=\sum_{k=1}^{K}\frac{1}{\sigma_t^2}\boldsymbol{\upsilon}_k^{*}\boldsymbol{\upsilon}_{k}^{T}\otimes \mathbf{I}_{L}+\frac{1}{\sigma_t^{2}}\boldsymbol{\upsilon}_{t}^{*}\boldsymbol{\upsilon}_{t}^{T}\otimes \mathbf{I}_{L},
% \end{align}
% \begin{align}
%     \mathbf{F}(\bar{\mathbf{g}}_{} )=\frac{1}{\sigma_t^2}(\mathbf{I}_{T} \otimes \boldsymbol{\Lambda}^{}\boldsymbol{\Lambda}_{}^{H}),
% \end{align}
\begin{align}
\label{eq:070}
\begin{split}
    \mathbf{F}(\dot{\bar{\mathbf{g}}}_{} )&=\frac{1}{\sigma_t^2}
    (\mathbf{I}_{T} \otimes \boldsymbol{\Lambda}_{}^{H})^H \cdot
    (\mathbf{I}_{T} \otimes \boldsymbol{\Lambda}_{}^{H})\\
    & \overset{(\text{III})}{=}\frac{1}{\sigma_t^2}
    (\mathbf{I}_{T} \otimes \boldsymbol{\Lambda}_{}^{}) \cdot
    (\mathbf{I}_{T} \otimes \boldsymbol{\Lambda}_{}^{H})\\
    & \overset{(\text{IV})}{=} \frac{1}{\sigma_t^2}
    (\mathbf{I}_{T} \otimes \boldsymbol{\Lambda} \boldsymbol{\Lambda}_{}^{H})\\
    & =  \frac{1}{\sigma_t^2}
    (\mathbf{I}_{T} \otimes \mathbf{R}_{{\Lambda}}),
\end{split}
\end{align}
where $\mathbf{R}_{{\Lambda}} := \boldsymbol{\Lambda} \boldsymbol{\Lambda}_{}^{H}$ is the signal covariance matrix. In addition, the step (III) in (\ref{eq:070}) is based on the property of $(\mathbf{A}\otimes \mathbf{B})^H=\mathbf{A}^H\otimes \mathbf{B}^H$, and the step (IV) in (\ref{eq:070}) follows the property that $(\mathbf{A}\otimes \mathbf{B})(\mathbf{C}\otimes \mathbf{D})=(\mathbf{AC}\otimes \mathbf{BD})$. 
We observe from the above formula (\ref{eq:070}) that: i) the trace of $\mathbf{R}^{-1}_{{\Lambda}}$, i.e., $\mathrm{Tr}(\mathbf{R}^{-1}_{{\Lambda}})$,
reflects the amount of information of the signal covariance matrix $\mathbf{R}_{{\Lambda}}$; ii)
as the rank of $\mathbf{R}_{{\Lambda}}$
increases (i.e., higher DoFs), $\mathrm{Tr}(\mathbf{R}^{-1}_{{\Lambda}})$ decreases, yielding a lower CRLB and better estimation performance. 
This completes the proof.

\end{appendices}

\bibliographystyle{IEEEtran}%
\bibliography{bibfile}
\end{document}